\documentclass[aps,prx,twocolumn,showpacs,superscriptaddress]{revtex4-1}

\usepackage{microtype}
\usepackage[colorlinks=true,bookmarks=false,citecolor=blue,urlcolor=blue]{hyperref} 

\usepackage{bbm}
\usepackage{amsmath}
\usepackage{hyperref}
\usepackage{graphicx}
\graphicspath{{figures/}} 
\usepackage{graphics}
\usepackage{amssymb}
\usepackage{mathtools}
\usepackage{xcolor}
\usepackage{physics}
\usepackage[british]{babel}
\usepackage{csquotes}
\usepackage{graphicx}
\usepackage{amsthm}
\usepackage{multirow}
\usepackage[normalem]{ulem}

\newcommand{\sfrac}[2]{{}^{#1}\mathclap{/}_{#2}}
\newcommand{\decode}{\mathcal{D}}
\newcommand{\code}{\mathcal{C}}
\newcommand{\hilbert}{\mathcal{H}}

\newtheorem{theorem}{Theorem}
\newtheorem{corollary}{Corollary}
\newtheorem{definition}{Definition}
\newtheorem{property}{Property}
\newtheorem{lemma}{Lemma}

\begin{document}

\title{Linear optical logical Bell state measurements with optimal loss-tolerance threshold}

\author{Paul Hilaire}
\affiliation{Huygens-Kamerlingh Onnes Laboratory, Leiden University, P.O. Box 9504, 2300 RA Leiden, The Netherlands}
\affiliation{Quandela, 7 Rue Léonard de Vinci, 91300 Massy, France}
\email{paul.hilaire@quandela.com}

\author{Yaron Castor}
\affiliation{Sorbonne Université, CNRS, LIP6, F-75005 Paris, France}

\author{Edwin Barnes}
\affiliation{Department of Physics, Virginia Tech, Blacksburg, VA 24061, USA}
\affiliation{Virginia Tech Center for Quantum Information Information Science and Engineering, Blacksburg, VA 24601, USA}

\author{Sophia E. Economou}
\affiliation{Department of Physics, Virginia Tech, Blacksburg, VA 24061, USA}
\affiliation{Virginia Tech Center for Quantum Information Information Science and Engineering, Blacksburg, VA 24601, USA}

\author{Frédéric Grosshans}
\affiliation{Sorbonne Université, CNRS, LIP6, F-75005 Paris, France}

\begin{abstract}
Quantum threshold theorems impose hard limits on the hardware capabilities to process quantum information. We derive tight and fundamental upper bounds to loss-tolerance thresholds in different linear-optical quantum information processing settings through an adversarial framework, taking into account the intrinsically probabilistic nature of linear optical Bell measurements. For logical Bell state measurements --- ubiquitous operations in photonic quantum information --- we demonstrate analytically that linear optics can achieve the fundamental loss threshold imposed by the no-cloning theorem even though, following the work of Lee et al., (Phys. Rev. A 100, 052303 (2019)), the constraint was widely assumed to be stricter. We spotlight the assumptions of the latter publication and find their bound holds for a logical Bell measurement built from adaptive physical linear-optical Bell measurements. We also give an explicit even stricter bound for non-adaptive Bell measurements.
\end{abstract}
\maketitle

\section{Introduction}
Photonic quantum technologies benefits from the advantages of photons: being in principe decoherence-free, traveling at light speed, and produceable at high clock rates. These advantages make photons arguably the medium of choice for quantum communications and a strong contender for quantum computing.

The simplest way to manipulate photonic quantum information is through linear optics.
A ubiquitous operation for photonic quantum information processing is the Bell state measurement (BSM)
\cite{Calsamiglia2001, Weinfurter1994, Braunstein1995, Michler1996, Lutkenhaus1999, Vaidman1999}, 
which is the cornerstone of fusion-based quantum computing~\cite{Bartolucci2021}, one of the current most advanced architectures for fault-tolerant linear-optic quantum computing.
Similarly, in quantum communications, they are critical for quantum teleportation \cite{Bennett1993, Vaidman1994, Boschi1998, Pirandola2015}
and for entanglement swapping \cite{Zukowski1993, Pan1998},
for example in quantum repeater protocols, whose purpose is to enable long-distance quantum communications~\cite{Briegel1998, Sangouard2011}. The most advanced quantum repeater schemes~\cite{Muralidharan2016, Jiang2009}, including the recently investigated all-photonic quantum repeaters~\cite{Azuma2015, Ewert2016, Lee2019,Hilaire2021error,Hilaire2021resource, Zhang2022, Niu2022, Bell2022}, are based on quantum error correction and logical BSMs. 

However, due to intrinsic limitations of linear optics, two-photon BSMs are inherently probabilistic with a success rate of at most $50\%$~\cite{Weinfurter1994, Braunstein1995, Michler1996, Lutkenhaus1999, Vaidman1999}. 
While photonic BSMs can be improved beyond this limit --- up to (near) determinism --- through the use of auxiliary resource states~\cite{Grice2011, Ewert2014, Wein2016, Olivo2018}, nonlinear interaction with an atom~\cite{Kim2001, Kim2002}, hyperentanglement~\cite{Kwiat1998, Walborn2003, Schuck2006, Barbieri2007} or squeezing~\cite{Zaidi2013}, these strategies have not been proved useful for loss tolerance.

Yet, we can use quantum error correction to make a logical BSM resistant to loss~\cite{Shor1995, Steane1996, Gottesman1997, Kitaev2003}. 
A quantum error correcting code (QECC) protects a quantum state from random qubit losses occurring with a probability below a threshold value, --- the loss-tolerance threshold --- which is intrinsically linked to the QECC used. 

Determining these loss-tolerance thresholds is crucial, and finding QECCs compatible with linear optics and which have the largest loss thresholds is of the utmost importance as it imposes a hard limit on the tolerable amount of loss of a quantum channel. In quantum communications, this threshold immediately translates into an upper bound on the distance between two nodes in quantum repeater schemes based on quantum error correction. For fault-tolerant quantum computing, it imposes an upper bound on the number of lossy operations that can be made, in a photonic quantum circuit, before the detection of photons. Moreover, the interest in loss-tolerant QECCs extends beyond the scope of photonics as Refs.~\cite{Wu2022, Kubica2022, Kang2022} have recently shown that we can more efficiently deal with some matter qubit errors by using erasure conversion, i.e. converting a qubit computational error into a heralded qubit loss.

In previous work, quantum error correction has already been adapted to a linear-optical setting. In~\cite{Varnava2006}, Varnava et al. proved that single-qubit logical measurements on a tree graph state QECC can be performed with $50 \%$ loss tolerance. In~~\cite{Ralph2005}, Ralph et al. showed how to protect a logical quantum state from loss using a quantum parity code and full linear-optical processing. Very recently, Bell et al.~\cite{Bell2022} have proposed new methods for QECCs based on graph states in a measurement-based setting, together with methods to improve their performances. This measurement-based quantum error correction has strong connections with linear-optical quantum information processing.
Logical BSMs have also been investigated, in particular because of their relevance for quantum communications.
Multiple all-photonic quantum repeater protocols, based on logical BSMs have been investigated~\cite{Sheng2015, Azuma2015, Ewert2016, Hilaire2021error, Niu2022}.
Lee et al.~\cite{Lee2019} upper bounded to $1-2^{-n}$ the success probability of a logical BSM acting on $n$-physical-qubit-encoded logical qubits through linear optical BSMs. This bound is reachable in the absence of losses.
In the presence of photon loss, they have also proposed an upper bound for the loss tolerance thresholds of logical BSMs based on linear optics and the no-cloning theorem. However, in a previous work~\cite{Hilaire2021error}, we have built linear-optical codes which worked around their assumptions and have numerically shown they overcome this bound.

In this paper, we derive fundamental and tight upper bounds on the loss-tolerance thresholds of logical BSMs both for general quantum information processing and when restricting ourselves to linear-optics. We derive these bounds using fundamental results of quantum physics such as the no-cloning theorem \cite{Wootters1982, Dieks1982}, the measurement postulate and by developing an adversarial framework based on quantum error correction. We then show that these fundamental bounds are actually achievable when restricting ourselves to a linear-optical measurement setting, by providing concrete  examples of QECCs with loss tolerance reaching these fundamental bounds. Furthermore, we use our framework to prove other linear-optical tight bounds for logical BSMs that can be implemented with less demanding technological requirements. These tight bounds also directly translate into a fully linear-optical loss-tolerant decoder based on quantum teleportation and with $50\%$ loss thresholds, the theoretically maximum achievable loss-tolerance.

This paper is organized as follows.
In Sec.~\ref{sec_qec}, we introduce briefly concepts of QECC and thresholds for qubit loss, and we present the principal components of our framework.
We investigate how no-go theorems of quantum mechanics impose bounds on these loss thresholds in Sec~\ref{sec_qec_bounds}, and we apply these results to logical BSMs in Sec.~\ref{sec_BSM}.
In Sec.~\ref{sec_LOBSM}, we focus on linear optics and investigate how this setting influences the loss thresholds for logical BSMs, investigating different logical requirements. In Sec.~\ref{subsec_fundamental_limit_lobsm}, we show that linear-optical BSMs has the same fundamental tight upper bound on loss tolerance as for general BSMs and we propose a logical decoder with the same loss threshold as the logical BSMs. Finally, we give an overview of the results and conclude in Sec.~\ref{sec_overview}.

\section{Quantum Error Correction} \label{sec_qec}

Here, we introduce the important concepts on QECCs that we will use in the following, to derive fundamental thresholds on logical linear-optical BSMs (LOBSMs). For a more detailed review on quantum error correction, we suggest Refs.~\cite{Gottesman1997, Lidar2013, Roffe2019}  to the interested readers.

\begin{figure}[tb]
    \centering
    \includegraphics[width=1 \columnwidth]{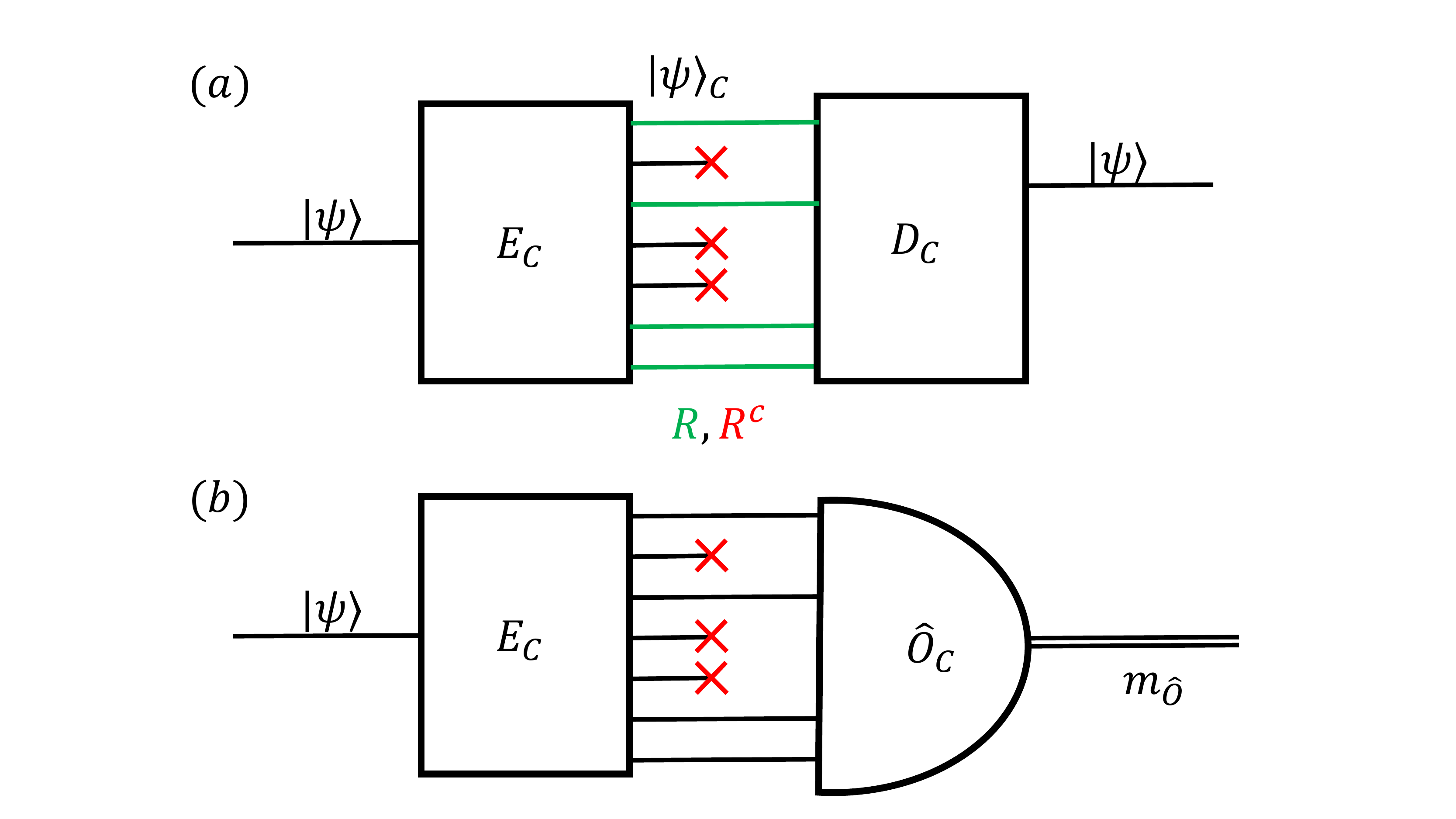}
    \caption{(a) Logical encoding of a quantum state using an encoder $E$ and a decoder $D$ to recover the state. Note that the decoder can still decide the state when it receives some qubit subsets $R \neq N$ (i.e.\ the $R^c$ physical qubit subset is lost).
    (b) Logical measurement of an operator $\hat O$, giving an $m_{\hat O}$ outcome.
    }
    \label{fig_logical}
\end{figure}

Quantum error correction is a strategy used to protect quantum information from errors where we encode one or more logical qubits onto many physical qubits.
While the strategy we present is likely adaptable to any kind of error QECCs can correct, we focus in this article on the erasure channel~\cite{Grassl1997} for simplicity reasons, keeping other error models for future work.
Formally, as illustrated in Fig.~\ref{fig_logical}, we can use a QECC $C$ to protect a state $\ket{\psi}$ by encoding it into a logical state $\ket{\psi_C}$, using its encoder $E_C$.
The logical quantum state is encoded onto a set $N$ of $|N|$ physical qubits.

\begin{definition}
    A perfect encoder $E_{C}$ of a QECC $C$ is an isometry mapping any quantum state $\ket{\psi} \in \hilbert_2^{\otimes |K|} $ describing a $|K|$-qubit state, into a subspace of a larger Hilbert space of $|N|\geq|K|$ qubits: $\dyad{\psi_C} = E_C(\ket{\psi})$ with $\ket{\psi_C}\in \hilbert_2^{\otimes |N|}$, following the encoding $C$.
\end{definition}

We follow the usual convention in the QECC litterature and assume the input state is unknown to the encoder.
Suppose that we only have access to a set $R$ of qubits, for example because the set of qubits $R^c = N \backslash R$ has been lost. In that case, we define an optimal decoder which aims at recovering the encoded quantum state $\ket{\psi}$ from this qubit subset $R$ (as illustrated in Fig.~\ref{fig_logical}):
\begin{definition}
    Let $C$ be a QECC that encodes any quantum state $\ket{\psi}$ onto a set $N$ of physical qubits as $E_C(\ket{\psi})$.
    An optimal decoder $D_C$ receives a subset $R \in N$ of physical qubits
    in the state $\Tr_{R^c}{E_C(\ket{\psi})}$. 
    It deterministically ouputs the $\ket{\psi}$ quantum state if it is theoretically possible to recover it from the subset $R$, or a failure flag $\ket\emptyset$ otherwise:
    $$
    D_C\left(\Tr_{R^c}{E_C(\ket{\psi})}\right)=
    \begin{cases}
        \ket{\psi} & R \in \decode_C,\\
        \ket\emptyset & {\rm otherwise},
    \end{cases}
    $$
    where we call $\decode_C$ the set of all decodable qubit subsets.
\end{definition}

\begin{lemma}
    For any QECC $C$, the probability $P(\decode_C|R)$ of recovering the quantum state from a qubit subset $R$, given the optimal decoder $D_C$ is an increasing function of $R$.
    \label{lemma_decode_r}
\end{lemma}

\begin{proof}
    According to the definition of an optimal decoder, $P(\decode_C|R) = \delta_{R \in \decode_C}$ with $\delta_{i} = 1$ (respectively $0$) if $i$ is true (false). If we can retrieve the quantum state from $R \in \decode_C$, we can also retrieve it from any $R' \supset R$. Therefore, $\forall R' \supset R, P(\decode_C|R') \geq P(\decode_C|R)$.
    \end{proof}

\subsection{Loss channel}

We focus here on loss channels that affect identically all the physical qubits from a QECC.
\begin{definition}
    A loss channel $C(\eta)$ is characterized by its loss probability, $\varepsilon$, or equivalently by its transmission efficiency $\eta = 1 - \varepsilon$, affecting each physical qubits. After passing through a loss channel $C(\eta)$ only a subset $R$ of the input physical qubit set $N$ is transmitted with probability, $P(N \to R | \eta)$:
    $$
    \forall R \subseteq N, \forall \eta  \in [0, 1], P(N \to R | \eta) = \eta^{|R|} (1 - \eta)^{|N|-|R|}.
    $$
\end{definition}

In the following, we will use $\varepsilon$ and $\eta$, to denote respectively loss and transmission efficiencies. 
A linear-optical implementation of a single-qubit lossy channel of transmission $\eta$ is a beamsplitter with the same transmission efficiency. 
Therefore, a lossy channel $C(\eta)$ can be also represented as a beamsplitter affecting each physical qubit as shown in Fig.~\ref{fig_channel_loss}(a).

\begin{figure}[!ht]
    \centering
    \includegraphics[width=.6\columnwidth]{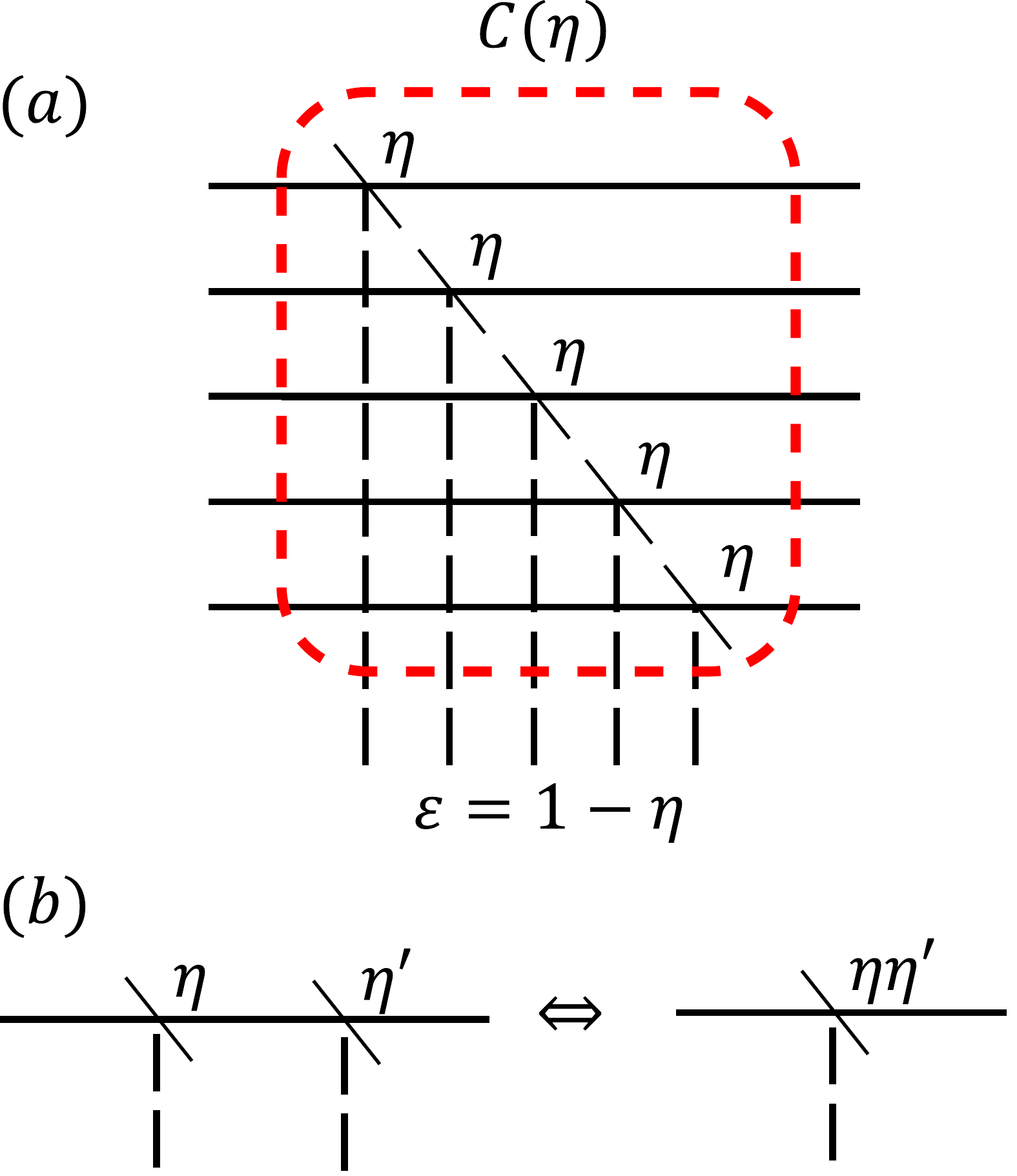}
    \caption{(a) Lossy channel $C(\eta)$ represented using multiple beamsplitters with reflectivity $\eta$. (b) Chain rule represented with beamsplitters.
    }
    \label{fig_channel_loss}
\end{figure}

\begin{property}
    The sum over all the possible output subsets $R$ of  the probabilities $P(N \to R|\eta)$ is unity:
    $$\forall \eta \in [0, 1], \forall N, \sum_{R \subseteq N} P(N \to R|\eta) = 1$$
    \label{prop_1}
\end{property}

\begin{property}     \label{prop_2}
    Chain rule of loss channels: $C(\eta_2 \eta_1) = C(\eta_2) \circ C(\eta_1)$, which gives the relation:
    \begin{multline}
         \forall \eta_1, \eta_2 \in [0, 1]^2, \forall N, \forall R \subseteq N, \nonumber \\
        P(N \to R|\eta_2 \eta_1) = \smashoperator{\sum_{R \subseteq R' \subseteq N}} P(N \to R' |\eta_1) P(R' \to R | \eta_2)  
    \end{multline}
\end{property}
Property~\ref{prop_2} is well understood in the beamsplitter representation of a lossy channel as illustrated inFig.~\ref{fig_channel_loss}(b).

Using this lossy channel, we can derive a corollary of Lemma.~\ref{lemma_decode_r}:

\begin{corollary}
    For all QECCs $C$, the probability to decode the quantum information given a physical qubit detection probability $\eta$, denoted as $P(\decode_C|\eta)$, is an increasing function of $\eta$.
\end{corollary}

\begin{proof}
    We first note that
\begin{align}
    P(\decode_C|\eta) &= \smashoperator{\sum_{R \subset N}} P(N \to R|\eta) P(\decode_C|R) \nonumber\\
    & = \smashoperator{\sum_{R \in \decode_C}} P(N \to R|\eta) \nonumber \\
    \intertext{Let us now consider, for each $R$, its supersets $R' \supseteq R$. For $\eta<\eta'$, }
    P(\decode_C|\eta)
    & = \smashoperator{\sum_{\substack{R, R' \in \decode_C \\ R \subseteq R' \subseteq N}}}  P(N \to R'| \eta') P(R' \to R| \eta/\eta')  \label{eq:RprimeOK} \\
    & = \smashoperator{\sum_{R' \in \decode_C}} P(N \to R'|\eta') 
        \Bigg(\smashoperator[r]{\sum_{\substack{R \in \decode_C \\ R\subseteq R'} }} P(R' \to R|\eta/\eta')\Bigg) \nonumber \\
    & \leq \smashoperator{\sum_{R' \in \decode_C}} P(N \to R'|\eta') = P(\decode_C|\eta'),   \label{eq:increasingeta}   
\end{align}
which ensures that $P(\decode_C|\eta)$ is increasing with $\eta$.
To obtain this result, we have used Lemma.~\ref{lemma_decode_r} to ensure that $R'$ is also  in $\decode_C$ in Eq,~(\ref{eq:RprimeOK}) and that the sum in parentheses is below 1 to obtain the inequality in Eq.~(\ref{eq:increasingeta}). 
\end{proof}

\subsection{Loss threshold}
We denote by $\code$ a family of QECCs, that is a (potentially infinite) set of QECCs.
An important result of quantum error correction is that there often exists  families of QECCs, $\code$,  which protect quantum information with an arbitrarily high success probability as long as the error rate is below some threshold value~\cite{Aharonov1997, Knill1998, Kitaev2003}. 
Interested by qubit losses, we define for each family of QECCs $\code$  the maximum qubit loss probability $\varepsilon_{\decode}$ such that we can use a code in $\code$ to correct it.

\begin{definition} \label{def_threshold}
    A family of QECCs, $\code$, has a loss threshold $\varepsilon_\decode$ if~{\normalfont \cite{Varnava2006,Stace2010}}:
\begin{equation*}
    \forall \varepsilon' > 0, \forall \varepsilon < \varepsilon_\decode, \exists C \in \code, P(\decode_C | 1 -\varepsilon) > 1 - \varepsilon',
\end{equation*} 
with $\decode_C$ the set of all decodable qubit subsets of a decoder $D_C$, and $\varepsilon$ the single physical qubit loss probability.
\end{definition}

One important remark is that for a specific QECC $C \in \code$, the probability $P(\decode_C | 1 -\varepsilon)$ is \emph{strictly below one} in the presence of losses $\varepsilon > 0$.  What this definition means is that, as long as the single-qubit loss is below a threshold value $\varepsilon_\decode$, it is always possible to find a QECC $C \in \code$, such that $P(\decode_C | 1 -\varepsilon)$ is arbitrarily close to 1 (but still strictly below 1).

\subsection{Measurements of observables}

We can also use a QECC to perform the measurement of an observable $\hat O$ which yields the measurement outcome $m_{\hat{O}}$ as shown in Fig.~\ref{fig_logical}.
We can decompose $\hat O$ with its spectral decomposition~\cite{Nielsen2000}: 
$$\hat O = \sum_{m_{\hat O}} m_{\hat O} \hat \Pi_{m_{\hat O}},$$
where $\hat \Pi_{m_{\hat O}}$ is the projector onto the $m_{\hat O}$-valued eigenspace of $\hat O$.

\begin{definition}
    We define a loss-tolerant threshold for a specific measurement $\varepsilon_{\hat O}$, similarly to the general loss-tolerant threshold:
    \begin{equation*}
        \forall \varepsilon' > 0, \forall \varepsilon < \varepsilon_{\hat O}, \exists C \in \code, P(\hat O_C | 1 - \varepsilon) > 1 - \varepsilon',
    \end{equation*} 
    where $P(\hat O_C | 1 - \varepsilon)$ is the probability to have the correct measurement outcome for the measurement of the logical operator $\hat O_C$, given a loss probability $\varepsilon$, using a code $C\in \code$.
\end{definition}

For practical reasons, we are often interested in the capability of the same code to allow the measurement of different observables in $\{O_k\}_k$, and/or to fully decode the state through $\decode$. 
Let $S$ be the set of operators we're interested in : $S=\{O_k\}_k$ or $S=\{\decode\}\cup\{O_k\}_k$.
We call $S$ a context and define loss tolerant thresholds $\varepsilon_i^S$, for each operators $i$ in the context $S$:
\begin{definition} \label{def_context}
    We define loss-tolerant thresholds for a specific context $S$, $\varepsilon_{i}^{S}$, $\forall i \in S$, such that:
    \begin{equation*}
        \forall \varepsilon' > 0, \forall \varepsilon < \varepsilon_{i}^S, \exists C \in \code, \forall i \in S,  P(i_C | 1 -\varepsilon) > 1 - \varepsilon'.
    \end{equation*} 
\end{definition}
Using that definition, we can for example investigate the loss-tolerance thresholds for the measurement of an operator $\hat O$ and $\decode$, $S=\{\hat O, \decode\}$. The derived loss thresholds $\varepsilon_{\hat O}^{S}, \varepsilon_{\decode}^S$ may now exhibit interdependence since we should find a unique code $C$ in $\code$ for which the conditions $P(\hat O |1 - \varepsilon) > 1 - \varepsilon'$ and $P(\decode_C |1 - \varepsilon) > 1 - \varepsilon'$ are simultaneously met.

The objective of this paper is to obtain these fundamental operator measurement thresholds in the case of BSMs built with different constraints.

\section{Fundamental limitations to quantum error correction} \label{sec_qec_bounds}

In the following, we show how fundamental theorems and axioms of quantum mechanics constrain the maximum performances of general QECCs.
All the results from this section are well established in the community, but we rederive them by introducing an adversarial framework where we consider that the qubit set $R^c$ that is lost through a lossy channel $C(\eta)$, is collected by an adversary. 
This framework is closely related to the seminal work of Cleve et al.\@ \cite{Cleve1999} on quantum secret sharing. 
We will later use this framework in a linear-optical setting to derive new fundamental results.
In the beamsplitter analogy for loss channels, this corresponds to actually collecting the physical qubits $R^c$ that were not transmitted through the channel, i.e.\@ beamsplitter (see Fig.~\ref{fig_channel_loss}).  

Note that the results that we derive here are completely general and are thus not limited to logical BSMs nor linear-optical implementations.

Indeed, while the loss thresholds necessarily depend on the family of codes considered, here, we are interested in fundamental and general upper bounds for the loss thresholds,
valid for any family of QECCs. The fundamental upper bounds derived in this section (and throughout this paper) are valid for \emph{any} QECCs.

\begin{figure}[!ht]
    \centering
    \includegraphics[width=1 \columnwidth]{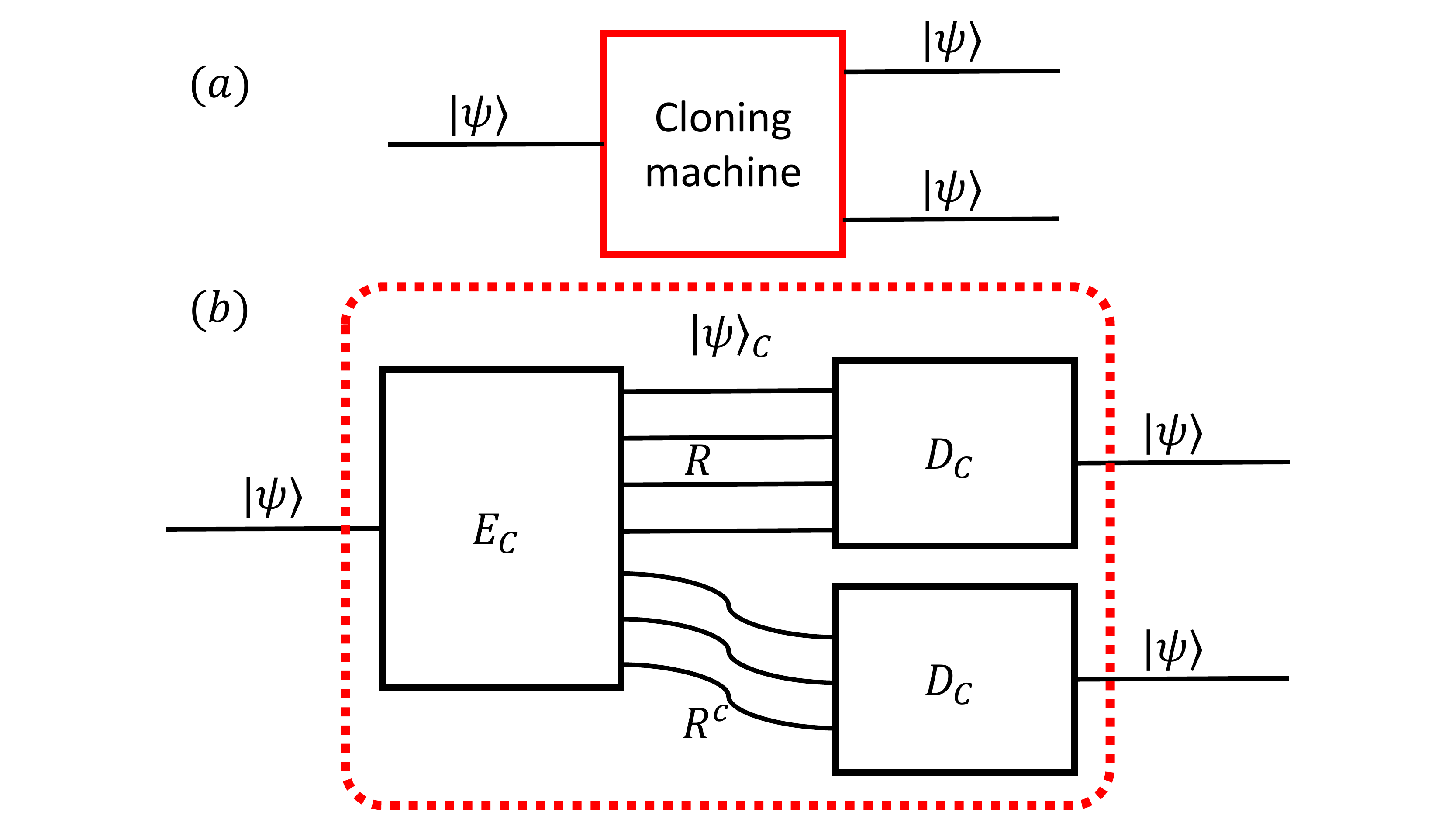}
    \caption{(a) Schematic of a non-physical cloning machine capable of reproducing any quantum states and (b) an implementation  using quantum error correction.
    }
    \label{fig_no_cloning}
\end{figure}

\subsection{No-cloning theorem}
We start by rederiving the known result  that the best loss-tolerance threshold for QECC is at most $\sfrac12$. This result originates from the no-cloning theorem~\cite{Wootters1982, Dieks1982} and will serve as a simple example of how to use our adversarial framework.
The theorem \ref{th_no_cloning_bound} is actually a straightforward extension of theorem 2 of \cite{Cleve1999} to
probabilistic losses and families of codes, and the beginning of its proof is taken from \cite{Cleve1999}.

The no-cloning theorem states that an unknown arbitrary quantum state cannot be cloned. Therefore, a ``cloning machine'' such as the one presented in Fig.~\ref{fig_no_cloning} is not physical and cannot be implemented. To understand the fundamental implications of the no-cloning theorem for QECC, we will devise a non-physical cloning machine based on QECC and use this impossibility result to derive the bound on the loss-tolerant threshold.

\begin{theorem} \label{th_no_cloning_bound}
    Due to the no-cloning theorem, the maximum amount of loss that any family of QECCs $\code$ can tolerate is strictly below $\sfrac12$:
    $$ \forall \code, \varepsilon_\decode \leq \sfrac12.$$
\end{theorem}

Note that, in Def.~\ref{def_threshold}, we have considered $\varepsilon < \varepsilon_\decode$, therefore, having the bound $\varepsilon_\decode \leq \sfrac12$ implies that it is impossible to decode a logical quantum state with arbitrarily high success probability in the presence of single-qubit loss probability, $\varepsilon$, \emph{equal or above} $50 \%$.

\begin{proof}
    
The non-physical cloning machine based on a QECC $C$ that we consider is represented in Fig.~\ref{fig_no_cloning}(b) and is composed of two decoders $D_C$. They respectively receive a subset $R$ and $R^c$ of the physical qubits. We know that such a cloning machine cannot work so that if $R \in \decode_C$, then $R^c \not\in \decode_C$.
We therefore have $$\forall C \in \code, \forall R \subset N, P(\decode_C | R) + P(\decode_C|R^c) \leq 1.$$

If we consider that each physical qubit has a probability $\eta = 1 -\varepsilon$ to go to the first decoder and $\varepsilon$ to go to the second decoder, we find that:
$$\forall C \in \code, \smashoperator{\sum_{R \subset N}} P(N \to R | \eta) (P(\decode_C | R) + P(\decode_C|R^c)) \leq 1,$$
since $\sum_{R \subset N} P(N \to R | \eta) = 1$.
Because $P(\decode_C | \eta) = \sum_{R \subset N} P(N \to R | \eta) P(\decode_C | R)$
 and $P(N \to R |\eta) = P(N \to R^c |1 - \eta)$, we therefore have
\begin{equation}
    \forall \varepsilon \in [0, 1], \forall C \in \code,  P(\decode_C|1 - \varepsilon) + P(\decode_C| \varepsilon) \leq 1
    \label{eq_no_cloning_th}
\end{equation}
Having $\varepsilon_\decode > \sfrac12$ for a class of QECCs $\code$ would imply, by Def.~\ref{def_threshold}  that there is a code $C \in \code$ such that  $P(\decode_C|\sfrac12)$ is arbitrarily close to one, say $P(\decode_C|\sfrac12)> \sfrac12$. This would contradict Eq~\eqref{eq_no_cloning_th}, and thus by contradiction $\varepsilon_\decode \leq \sfrac12$.

There actually exist QECCs for which we reach this loss-tolerance threshold limit such as surface codes~\cite{Kitaev2003, Stace2009}. Therefore, the $\varepsilon_\decode \leq \sfrac12$ bound is tight.  
\end{proof}

\subsection{Measurement postulate}

\begin{figure}[!htb]
    \centering
    \includegraphics[width=1\columnwidth]{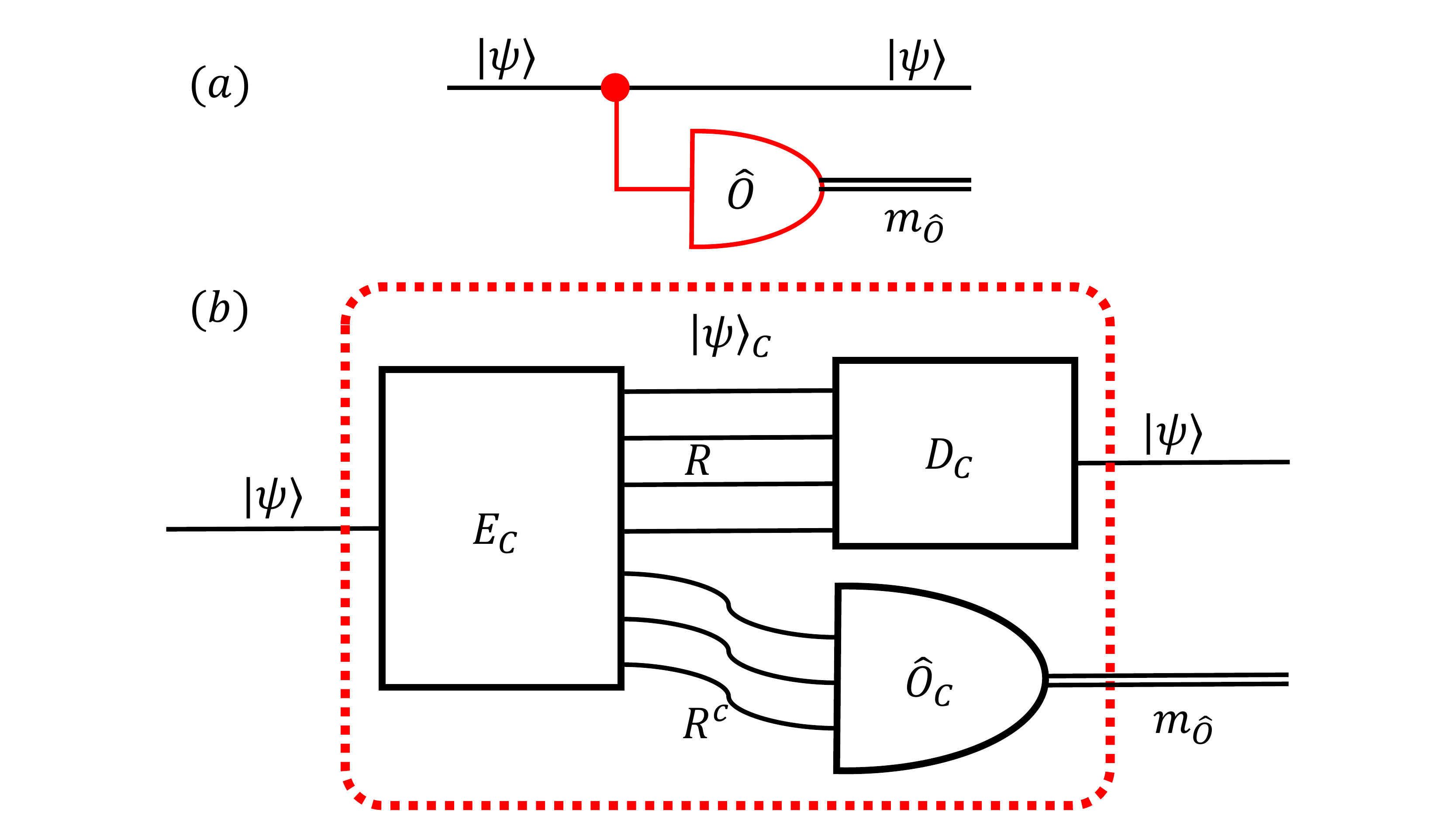}
    \caption{(a) Non-physical non-projective measurement machine capable of measuring the operator $\hat O$ while recovering the arbitrary quantum state $\ket{\psi}$ and (b) an implementation using quantum error correction.
    }
    \label{fig_no_measurement}
\end{figure}

\begin{figure}[!htb]
    \centering
    \includegraphics[width=1\columnwidth]{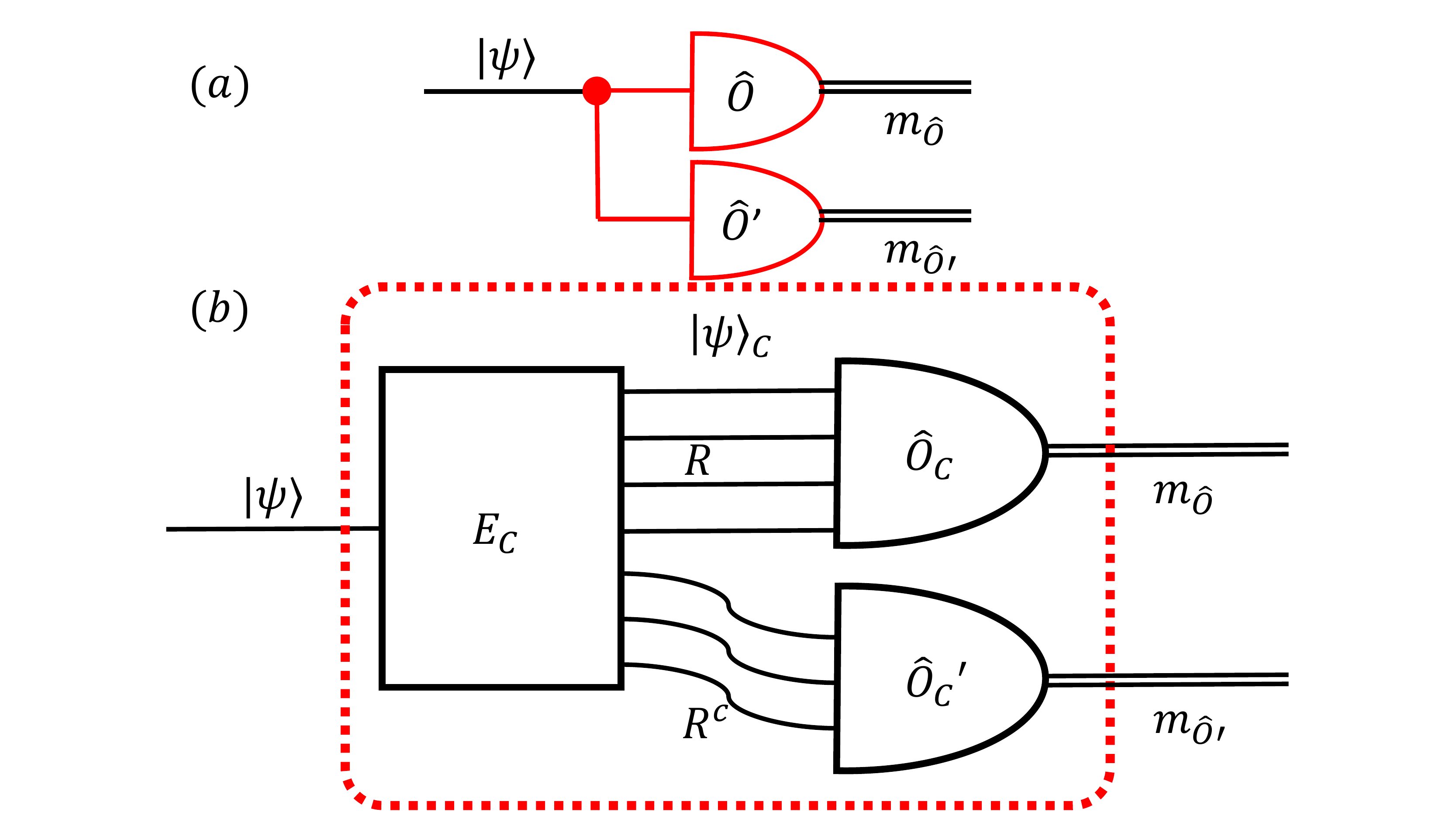}
    \caption{(a) Non-physical measurement machine capable of measuring non-commuting operators $\hat O$ and $\hat O'$ and (b) an implementation using quantum error correction.
    }
    \label{fig_no_measurement_2}
\end{figure}

The measurement postulate of quantum mechanics states that, after the measurement of an observable $\hat O$, the system's wave-function collapses, and (i) cannot recover the initial state $\ket{\psi}$. 
Moreover, (ii) if two observables $\hat{O}$, $\hat O'$ do not commute ---$[\hat{O}, \hat O']\neq0$---, they cannot be measured simultaneously.  
As a result measurement apparata such as the ones depicted in Fig.~\ref{fig_no_measurement} and \ref{fig_no_measurement_2} are not physically allowed, which immediately translates into bounds on loss tolerance.
If a decoder receives a subset $R$ of the physical qubits and the measurement device for $\hat O$ logical measurements receives the $R^c$ qubits, given (i) they cannot simultaneously succeed. This leads to the result:

\begin{theorem}
    The measurement postulate implies two bounds on the respective loss tolerance of the decoder and operator measurement devices:
    \begin{equation}
        \forall \code, \forall \hat O, \varepsilon_\decode^{\{\decode, \hat O\}} + \varepsilon_{\hat O}^{\{\decode, \hat O\}} \leq 1.
        \label{eq_code_measurement}
    \end{equation}
    and 
    \begin{equation}
        \forall \code,  \forall [\hat O, \hat O'] \neq 0, \varepsilon_{\hat O}^{\{\hat O, \hat O'\}} + \varepsilon_{\hat O'}^{\{\hat O, \hat O'\}} \leq 1.
        \label{eq_measurement_comm}    
    \end{equation}
\end{theorem}

\begin{proof}[Proof of Eq.~\eqref{eq_code_measurement}]
    Since a circuit such as the one depicted in Fig.~\ref{fig_no_measurement} is not physical, we can follow the
    reasoning in the proof of Theorem \ref{th_no_cloning_bound}:
    \begin{equation}
        \forall \varepsilon> 0, \forall C \in \code, P(\decode_C|1 - \varepsilon) + P(\hat O_C |\varepsilon) \leq 1.
        \label{eq_single_code_measurement}
    \end{equation}

    If we consider a class of QECCs $\code$ which has a loss-tolerance threshold $\varepsilon_\decode^{\{\decode, \hat O\}}$, we can find codes $C \in \code$ for which $P(\decode_C|1- \varepsilon) \to 1$, $\forall \varepsilon < \varepsilon_\decode^{\{\decode, \hat O\}}$. 
    Therefore, for each of these codes $C$, Eq.~(\ref{eq_single_code_measurement}) leads to 
    $1 - P(\decode_C|1- \varepsilon) \ge P(\hat O_C | \varepsilon) \to 0$ and consequently 
    $1 - \varepsilon_{\decode}^{\{\decode, \hat O\}} \geq \varepsilon_{\hat O}^{\{\decode, \hat O\}}$.
    We therefore have proved the first inequality Eq.~\eqref{eq_code_measurement} between the loss-tolerant threshold and the loss-tolerant $\hat O$ measurement threshold.
\end{proof}
    
    An important subtlety is that the loss-tolerance threshold for a measurement is not strictly bounded by $\sfrac12$,
    as it is the case for a QECC loss-tolerant decoder, but by $1 - \varepsilon_\decode^{\{\decode, {\hat O}\}}$ which can be arbitrarily close to $1$ if we consider a class of codes for which $\varepsilon_\decode$ is arbitrarily close to $0$.

    For example, the class of repetition codes, for which $E(\ket{i})=\dyad{i_C} = \dyad{i}^{\otimes n}$ (for $i = 0, 1$), is not loss-tolerant $\varepsilon_\decode = 0$. Yet, we can make a fully loss-tolerant logical $Z$ measurement by taking $n$ sufficiently large, 
    by simply measuring $Z$ on every qubit. 
    Hereafter, $X$, $Y$, $Z$ denotes the usual Pauli operators. Indeed, a logical $Z$ measurement will succeed if at least one physical $Z$ measurement on any physical qubits of the collected subset $R$ succeeds, which occurs with a non-zero probability if $\varepsilon < 1$. Therefore, as long as $\varepsilon < \varepsilon_{Z} = 1$, we can loss-tolerantly measure the $Z$ operator. This is the reason why we are considering the context in the derivation of loss-tolerance thresholds in Def.~\ref{def_context}.

\begin{proof}[Proof of Eq.~\eqref{eq_measurement_comm}]
    The proof of the second inequality (Eq.~\eqref{eq_measurement_comm}) is based on (ii) from the measurement postulate and follows exactly the same reasoning as for the proof of Eq.~\eqref{eq_code_measurement} using the non-physical machine represented in Fig.~\ref{fig_no_measurement_2} and
    \begin{multline}\label{eq_p_oo}
         \forall \varepsilon \in [0, 1], \forall [\hat O, \hat O'] \neq 0,\\   P(\hat O_C |1 - \varepsilon) + P(\hat O'_C |\varepsilon)\leq 1.
    \end{multline}
\end{proof}

These results are a consequence of the measurement postulate applied to QECCs and Eq.~(\ref{eq_p_oo}) has strong connections with quantum secret sharing~\cite{Hillery1999, Cleve1999, Gottesman2000, Markham2008} and conjugate coding~\cite{Wiesner1983}.

For ``good'' loss-tolerant QECCs for which $\varepsilon_{\decode} = \sfrac12$, it follows from these previous results (Eq.~\eqref{eq_code_measurement} and Eq.~\eqref{eq_measurement_comm}) that we can have at best $\varepsilon_\decode^{S} = \varepsilon_{{\hat O}}^{S} = \varepsilon_{{\hat O'}}^{S} = \sfrac12$, for $S = \{\decode, \hat O, \hat O'\}$. This limit is actually tight since we can, for example, perform loss-tolerant Pauli $X$ and $Z$ measurements ($[X, Z] \neq 0$) onto the class of surface codes and show that these measurements have a $50\%$ loss-threshold due to percolation theory~\cite{Stace2009}.

\begin{figure}[!ht]
    \centering
    \includegraphics[width=1 \columnwidth]{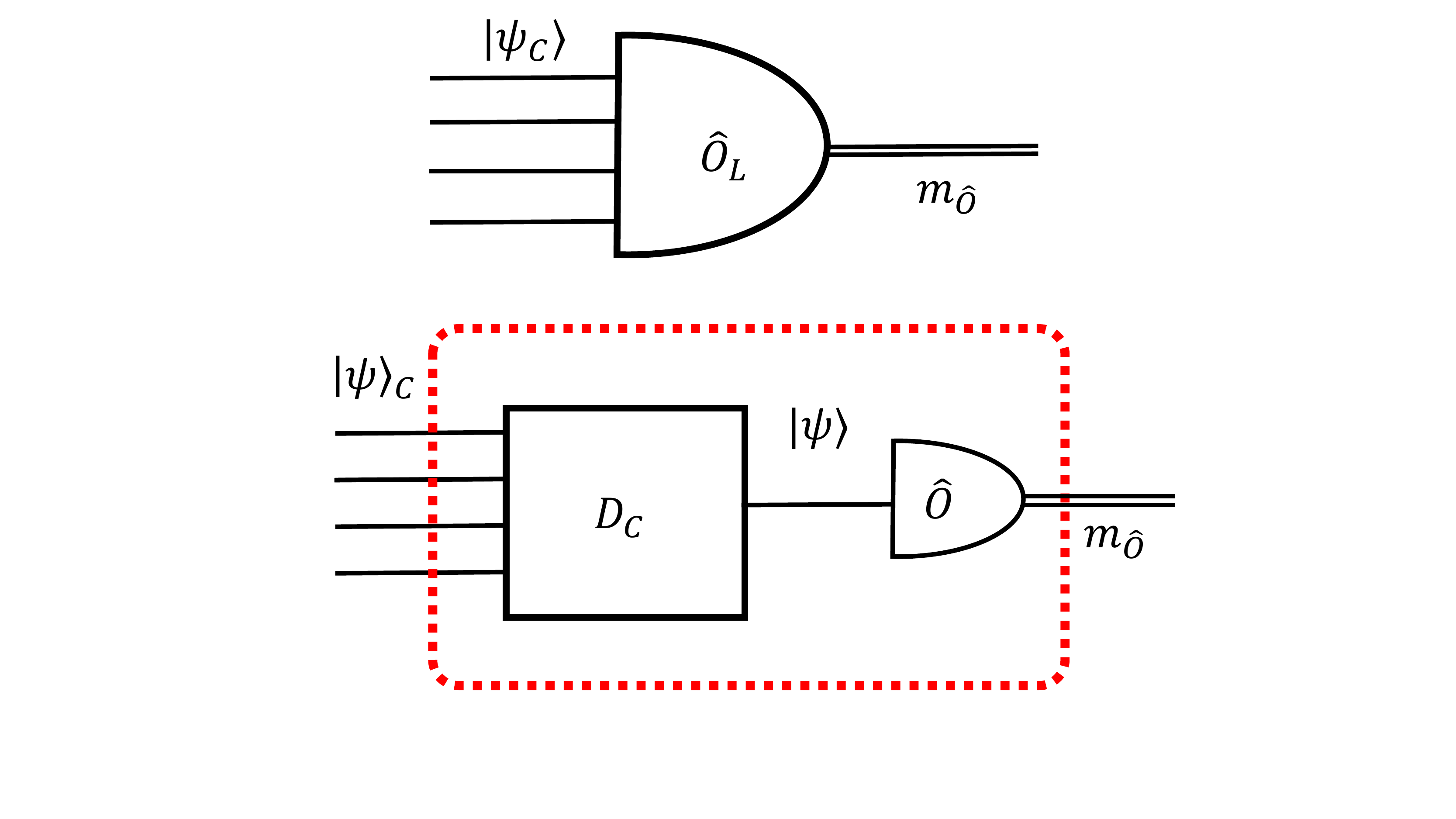}
    \caption{Logical operator measurements $m_{\hat O}$ made using the decoder $D_C$ and a physical operator measurement $\hat O$.
    }
    \label{fig_no_measurement_3}
\end{figure}

\begin{lemma}
    For all families of codes $\code$, if a context $S$ includes the decoder $\decode$ and operator measurements, then 
    $\varepsilon_\decode^S$ is the smallest threshold of this context: 
    $$\forall i \in S, \varepsilon_\decode^S \leq \varepsilon_i^S.$$
\end{lemma}

\begin{proof}
    We could perform any logical measurement $\hat O$, by first decoding the logical qubit state and then performing the physical measurement onto the physical qubit, as shown in Fig.~\ref{fig_no_measurement_3}. Therefore, the loss tolerance of any QECC is always smaller than or equal to the loss tolerance of any measurement.
\end{proof}

\section{Fundamental bounds on logical Bell-state measurements}\label{sec_BSM}

\subsection{Physical and logical Bell-state measurements}

A Bell-state measurement (BSM) on two physical qubits $a$ and $b$ projects them into one of these four Bell states:
\begin{equation}
\begin{aligned}
    \ket{\Phi^{\pm}_{a,b}} & = \frac{1}{\sqrt{2}} \left( \ket{0_a, 0_b} \pm \ket{1_a, 1_b} \right) & \Leftrightarrow & \left\{ \substack{Z_a Z_b \to +1 \\ X_a X_b \to \pm 1} \right. \\
    \ket{\Psi^{\pm}_{a,b}} & = \frac{1}{\sqrt{2}} \left(\ket{0_a, 1_b} \pm \ket{1_a, 0_b} \right) & \Leftrightarrow & \left\{ \substack{Z_a Z_b \to -1 \\ X_a X_b \to \pm 1} \right. 
\label{eq_bsm_stab}
\end{aligned}
\end{equation}

As shown in Eq.~(\ref{eq_bsm_stab}), a BSM is equivalent to the joint measurement of the two stabilizer operators $Z_a Z_b$ and $X_a X_b$.
Here, we use a subscript $a$ or $b$ to denote the qubit the Pauli operator is acting on.
A BSM is thus successful if we have successfully measured these two stabilizers. Rewriting a BSM into stabilizer measurements will be essential to make a link with quantum error correction in the following.
We emphasize that the results that we will derive in the following are completely general, 
and we will focus  on linear optics only from Section~\ref{sec_LOBSM} onwards.

\subsection{Arbitrarily high loss-tolerant BSM threshold}

\begin{figure}[!htb]
    \centering
    \includegraphics[width=1\columnwidth]{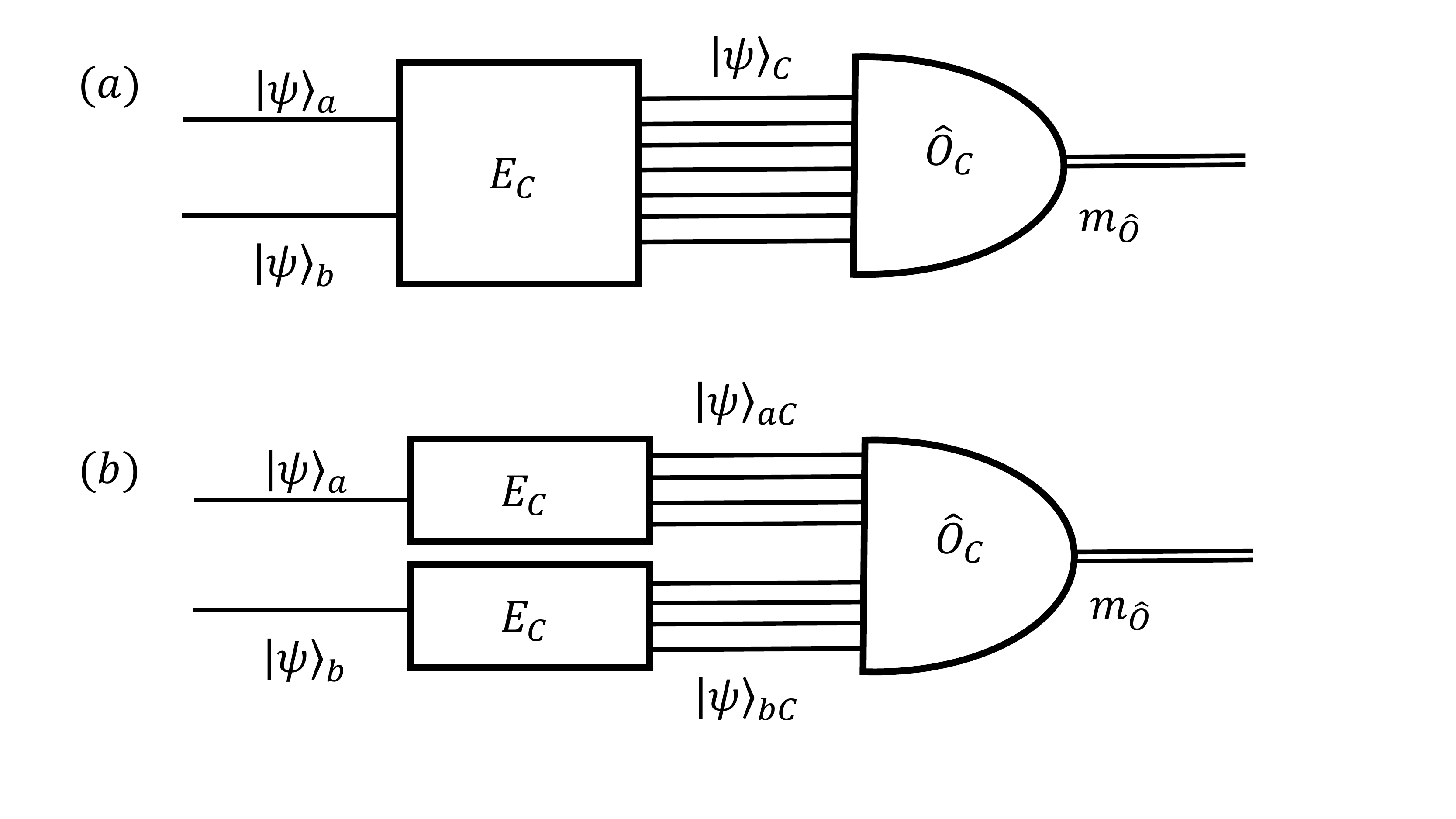}
    \caption{ (a) General logical encoding of a two-party quantum state shared by parties $a$ and $b$ and a measurement apparatus allowing the measurement of the operator set $\hat O$. (b) Variant where the two parties realize the logical encoding remotely.}
    \label{fig_logical_2}
\end{figure}

In the following, we consider the fundamental loss tolerance of logical BSMs in the two different settings displayed in Fig.~\ref{fig_logical_2}: in the general case, and when the logical encoding of the two qubits $a$ and $b$ is local. 
A logical BSM acting on two logical qubits $a$, $b$ (using a given QECC $C$) corresponds to the joint measurement of the operator set $\{X_{aC} X_{bC}, Z_{aC} Z_{bC}\}$, where we have added the subscript $C$ to indicate that it is a logical operator acting on 
the logical space defined by the QECC $C$.

\begin{lemma}
    For codes prepared as in Fig.~\ref{fig_logical_2}(a), the fundamental logical BSM loss-tolerance threshold is only upper-bounded by unity: $\varepsilon^{(BSM)} \leq 1$.
\end{lemma}
    This surprising first result seems promising --- we can find QECCs for which $\varepsilon^{(BSM)} = 1$ ---, 
    however, these codes are not as useful as it appears and this is essentially a consequence of the power of non-local encoding.
\begin{proof}
    Because $\{X_{aC} X_{bC}, Z_{aC} Z_{bC}\}$ is a set of commuting operators, we do not have limitations such as the one derived in Eq.~\eqref{eq_measurement_comm} and we are thus only limited by Eq.~\eqref{eq_code_measurement}. We find a variant of the class of repetition codes for which we can make arbitrarily loss-tolerant BSMs: 
    since the Bell states are orthogonal, we can encode two logical qubits into $2n$ physical qubits in such a way that
    logical Bell states are given by $\ket*{\mathrm{Bell}_C^{(n)}} = \ket{\mathrm{Bell}}^{\otimes n}$, 
    with $\ket{\mathrm{Bell}} \in \{ \ket{\Phi^\pm}, \ket{\Psi^\pm}\}$. 
    We can perform a logical BSM by performing physical two-qubit BSMs, only one of which needs to succeed. This logical BSM succeeds with arbitrarily high probability for sufficiently large $n$ for any $\varepsilon < 1$. Therefore, for this family of QECCs, $\varepsilon^{(BSM)} = 1$, yet similarly to the usual repetition code they are not loss-tolerant ($\varepsilon_{\decode}= 0$), as expected from Eq.~\eqref{eq_code_measurement}. We can also show that this threshold also holds when we restrict ourselves to a linear-optical setting.    
\end{proof}

This surprising result deserves some discussions. First, we should note that the encoding requires that the encoder $E$ has access to the two parties $a$ (having access to the state $\ket{\psi_a}$) and $b$  (having access to the state $\ket{\psi_b}$) to prepare the logical state (see Fig.~\ref{fig_logical_2}(a)). Therefore, an interesting open question is whether this result is of practical interest for quantum computing and quantum communications, which usually consider BSM performed on states prepared locally by different parties. 

\subsection{Loss-tolerant threshold of BSM on locally-prepared QECCs}

A more practical setting --- which we will consider in the remainder of this paper --- is the one depicted in  Fig.~\ref{fig_logical_2}(b). Here, contrary to the general case, each party $a$ and $b$ performs their logical encoding locally. Consequently, the physical qubits encoding the logical qubits owned by $a$ are not the same as the physical qubits owned by $b$. Such a configuration is particularly important for example in quantum communications, where the logical BSMs are performed on logical qubits generated remotely at different nodes. The class of ``Bell repetition codes'' cannot be prepared locally by each party and, consequently, the loss-tolerance threshold under such a restriction may be --- and actually is --- different, and we will call it $\bar \varepsilon^{(BSM)}$.

\begin{theorem}\label{th_bsm}
    Logical BSMs performed on locally-encoded logical qubits have a loss-tolerance threshold of at most $\sfrac12$:
    $$\forall \code, \bar \varepsilon^{(BSM)} \leq \sfrac12.$$
\end{theorem}

\begin{figure}[!ht]
    \centering
    \includegraphics[width=1\columnwidth]{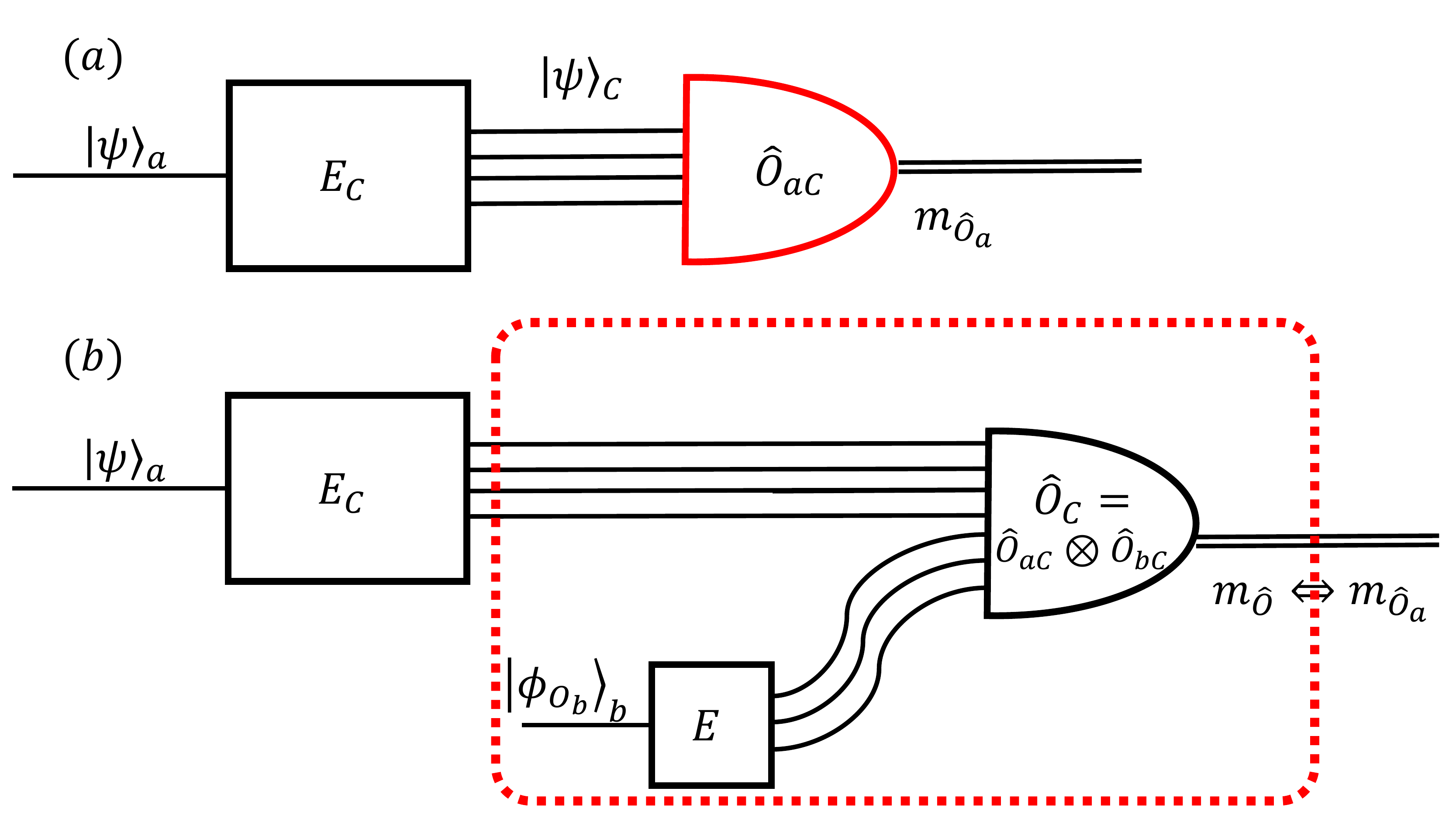}
    \caption{ (a) Logical measurement of the operator $\hat O_a$. (b) Implementation using a two-party logical encoding with a prepared state $\phi_{\hat 0_b}$ generated by $b$ and by measuring $\hat O_C = \hat O_{aC} \otimes \hat O_{bC}$.}
    \label{fig_logical_Oa}
\end{figure}

\begin{proof}
    To find this theoretical bound, we consider how to make single-qubit logical measurements of an operator $\hat O_a$ on party $a$, using the second party $b$ as an ancilla, as depicted in Fig.~\ref{fig_logical_Oa}. 
    Here, we use a known state $\ket*{\phi_{\hat O_b}}$,
    chosen to be an eigenstate of a second logical operator $\hat{O}_b$ acting on system $b$,
    so that the measurement of this operator yields a deterministic outcome $m_{\hat O_b}$. We then measure $\hat O = \hat O_a \otimes \hat O_b$ such that the measurement outcome $m_{\hat O}$ corresponds to an indirect measurement of $m_{\hat O_a}$. For example, if $\hat O = X_a X_b$ and $\ket*{\phi_{\hat O_b}} = \ket{\pm}_b$, then $m_{X_b} = \pm 1$, and $m_{X_a} = m_{X_a X_b} m_{X_b}$. 
    We denote the threshold for successfully measuring $\hat{O}_a$ in this indirect way by $\bar{\varepsilon}_{\hat{O}_a}$, and we similarly define $\bar{\varepsilon}_{\hat{O}_b}$. Since this indirect approach is a way to measure $\hat{O}_a$, it can only work provided $\varepsilon_a<\varepsilon_{\hat{O}_a}$, where $\varepsilon_a$ is the single-qubit loss rate of the physical qubits in $a$, and $\varepsilon_{\hat{O}_a}$ is the threshold for successfully measuring $\hat{O}_a$ regardless of the method used. We similarly denote the loss rate of the physical qubits in $b$ by $\varepsilon_b$, and in general $\varepsilon_a\neq\varepsilon_b$.
    By switching the roles of $a$ and $b$, we find the same result for $\varepsilon_b$: $\varepsilon_b < \varepsilon_{\hat O_b}$.
    We therefore find that the two loss-tolerance thresholds $\bar \varepsilon_{\hat O, a}$, $\bar \varepsilon_{\hat O, b}$
    for measuring the operator $\hat O = \hat O_a \otimes \hat O_b$ are bounded by:
    \begin{align}
        \bar \varepsilon_{\hat O, a} &\leq \varepsilon_{\hat O_a} ;     &\bar \varepsilon_{\hat O, b} &\leq \varepsilon_{\hat O_b}. 
    \end{align}
    This should hold if we measure $X_a$ through a $X_a X_b$ measurement and $Z_a$ through a $Z_a Z_b$ measurement, as above. 
    Consequently, since a BSM is a joint measurement of $X_a X_b$ and $Z_a Z_b$, we should have
    \begin{align}
        \bar \varepsilon^{(BSM)}_{a} & \leq \min [\bar \varepsilon_{X_a X_b, a}, \bar \varepsilon_{Z_a Z_b, a}] \nonumber \\ & \leq \min[\varepsilon_{X_a}^{\{X_a,Z_a\}}, \varepsilon_{Z_a}^{\{X_a,Z_a\}}] \leq \sfrac{1}{2}, \label{eq_a_bar_eps_a}\\
        \bar \varepsilon^{(BSM)}_{b} &  \leq \min [\bar \varepsilon_{X_a X_b, b}, \bar \varepsilon_{Z_a Z_b, b}]  \nonumber \\ & \leq \min[\varepsilon_{X_b}^{\{X_b,Z_b\}}, \varepsilon_{Z_b}^{\{X_b,Z_b\}}] \leq \sfrac{1}{2},
    \end{align}
    which conclude the proof.
    In Eq.~\eqref{eq_a_bar_eps_a}, we go from the first to the second line by using the fact that a $Z_a Z_b$ (respectively $X_a X_b$) measurement can be used to perform an indirect $Z_a$ (resp. $X_a$) measurement: $\bar \varepsilon_{Z_a Z_b, a} \leq \varepsilon_{Z_a}$ (resp. $\bar \varepsilon_{X_a X_b, a} \leq \varepsilon_{X_a}$). 
    However, the minimization should be done jointly onto both $X_a X_b$ and $Z_a Z_b$, which implies that the minimization on the single-qubit operators $X_a$ and $Z_a$ should be done in the same context $\{X_a, Z_a\}$, hence the result in the second line.
    Then, we have obtained the final result by maximizing the loss-thresholds for the operators $X_a$ and $Z_a$ (respectively for $X_b$ and $Z_b$) under the constraint Eq.~\eqref{eq_measurement_comm}, because they are non-commuting.
\end{proof}
We should also note that we have considered here independent loss thresholds for each party $a$ or $b$. 
Therefore, we have derived bounds on locally-encoded logical BSMs based on single-qubit measurement loss-thresholds. 
We emphasize that the loss-thresholds for logical BSMs derived in this section are very general and do not depend on physical implementation.
In Sec.~\ref{subsec_fundamental_limit_lobsm}, we will show that this bound is actually tight even when we restrict ourselves to a linear-optical setting.

\section{Fundamental bounds on linear-optical logical BSM} \label{sec_LOBSM}

\subsection{Linear-optical quantum information processing}
\label{sec_loqip}

In the remainder of this paper, we will focus on linear-optical BSMs.
In linear-optical quantum information processing, photons are the quantum information carriers, and we process them by using optical interferometers for unitary transformations, and photon-number-resolving detectors for measurements. 
We assume that physical qubits are encoded onto single photons using a dual-rail encoding, which is the most common encoding for discrete variable photonic quantum information processing.

Linear-optical elements enable arbitrary unitary transformations between photonic modes, yet are not sufficient to perform  deterministic two-qubit gates between photonic qubits. Two-qubit gates are thus inherently probabilistic using linear-optical quantum information processing~\cite{Knill2001}. For this reason, it is hard  (though not impossible~\cite{Ralph2005}) to re-encode a logically-encoded qubit deterministically. 
 
Another limitation concerns the qubit measurements built using single-photon or potentially photon-number resolving detectors. 
These detectors are intrinsically destructive measurement apparata so that the detected photons do not exist anymore after their
measurements, and thus cannot be reused. Such a subtlety is generally well handled in both a quantum communication setting and in linear-optics quantum computing such as the measurement-based or fusion-based quantum computing paradigms~\cite{Raussendorf2001, Raussendorf2006, Bartolucci2021, Lee2023}.

A LOBSM without ancillary photons has an intrinsic success probability of at most $\sfrac12$~\cite{Calsamiglia2001}. In practice, it is possible to design a linear optical setup \cite{Weinfurter1994, Braunstein1995, Michler1996} which measures the operator $Z_a Z_b$ measurement deterministically but only realizes the $X_a X_b$ measurement for a particular measurement outcome of the $Z_a Z_b$ operator, either +1 or -1, which occurs with probability $\sfrac12$. The role of $Z_a Z_b$ and $X_a X_b$ can be switched, and we can design other LOBSM setups which measure $X_a X_b$ deterministically and $Z_a Z_b$ with success probability $\sfrac12$.

\begin{figure}[!tb]
    \centering
    \includegraphics[width=1 \columnwidth]{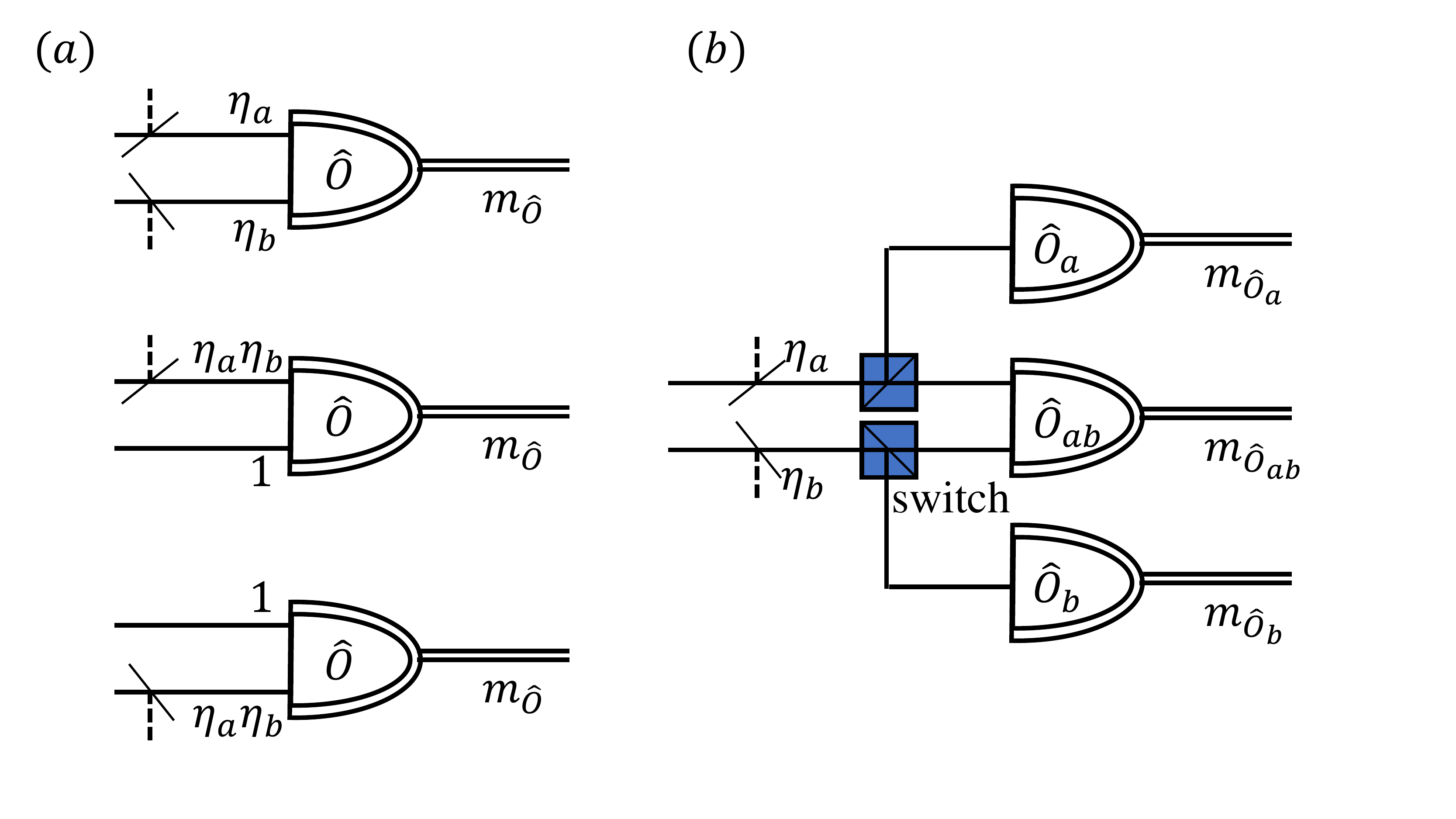}
    \caption{(a) Top: Physical two-photon linear-optical measurement apparatus measuring the operator $\hat O$. The detection probabilities $\eta_a, \eta_b$ of each photon yield a success probability $\eta_a \eta_b$. Middle and bottom: Each of these photon detection probabilities yields the same measurement success probability. (b) Measurement apparatus which can switch on demand between two-photon joint measurements and two single-photon measurements.}
    \label{fig_LOBSM}
\end{figure}

We now focus on logical BSMs built with linear optics. Figure~\ref{fig_LOBSM}
displays the measurement devices that we are considering for LOBSM at the physical level. Such a device can only yield a successful result if the two qubits are successfully detected, which occurs with a probability $\eta_a \eta_b$, $\eta_a$ and $\eta_b$ being the single-photon detection probabilities for each QECC. It is also based on the interference of the two photonic qubits which erase the which-path information. Consequently, it doesn't matter how the losses are spread over the two physical qubits: if qubit $a$ has a detection probability $\eta_a' = \eta_a \eta_b$, while $\eta_b' = 1$, the LOBSM device has the same success probability. Therefore, the three configurations illustrated in Fig.~\ref{fig_LOBSM} have the same success probability. Besides, a LOBSM yields a result probabilistically upon two-photon detection, such that we can measure $X_a X_b$ with a probability $p_{XX}$ and $Z_a Z_b$ with probability $p_{ZZ}$ with $p_{XX} + p_{ZZ} \leq \sfrac32$, and the probability to measure both of them is $p_{XX \cup ZZ}\leq \sfrac12$. 

Therefore, we have three constraints for a LOBSM:
\begin{enumerate}
    \item Joint successful detection probability ($\eta_a \eta_b$). \label{LOBSMc1}
    \item $p_{XX \cup ZZ} \leq p$  \label{LOBSMc2}
    \item $p_{XX} + p_{ZZ} \leq 1 + p$.  \label{LOBSMc3}
\end{enumerate}
The overall success probability of a LOBSM including losses is therefore $P_{LOBSM} = \eta_a \eta_b p_{XX \cup ZZ}$, while we can recover the $XX$ ($ZZ$) component with probability $P_{XX} = \eta_a \eta_b p_{XX}$ ($P_{ZZ} = \eta_a \eta_b p_{ZZ}$). Without the use of  auxiliary states, we generally cannot do better than having  $p = \sfrac12$, but the intrinsic success probability of a LOBSM can be made near-deterministic $p \to 1$ through the use of many entangled ancillary photons \cite{Grice2011}.
This solution is in practice challenging to implement since a successful LOBSM would generally require all the ancillary photons to be detected and thus would make such a LOBSM much more sensitive to photon losses.

By restricting our logical measurement device so that we can implement it using linear optics, the question that we are now addressing is whether we can reach the fundamental bounds that we already derived in the general case. 
While doing, we also find stricter bounds for more restricted constraints, in subsections \ref{sec_adapt_LOBSM} and 
\ref{subsec_static_LOBSM}, the former one being already discovered by Lee et al.\@ \cite{Lee2019}, 
who thought it applied to all linear optical setups.

\subsection{Transverse logical operator measurements}

\begin{figure}[!htb]
    \centering
    \includegraphics[width=0.7\columnwidth]{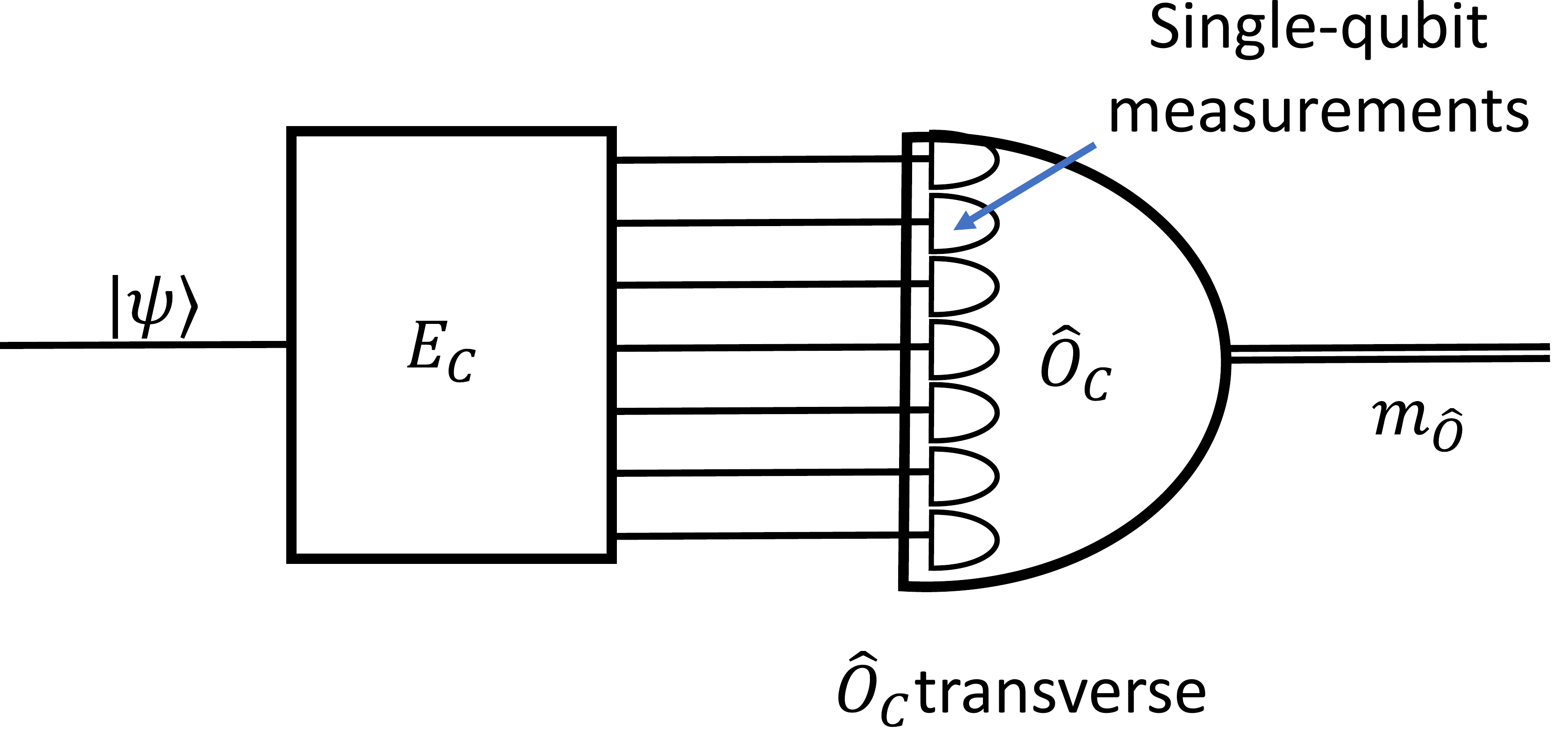}
    \caption{Transverse logical operator $\hat O_C$ which can be measured using a decoder made of single-qubit operations.}
    \label{fig_transverse}
\end{figure}

In the following, we will focus on single- and two-qubit logical operators which are transverse. We use the definition of a transverse operator, $\hat O$, as a logical operator which can be decomposed into single-qubit operators $\hat O = \hat O_1 \otimes \hat O_2 \otimes \cdots \otimes \hat O_{|N|} $. As a result, we can also perform a logical measurement of $\hat O$ by performing only individual physical single qubit measurements as illustrated in Figure~\ref{fig_transverse}. For our concerns, individual single qubit measurements are interesting since they can be natively implemented in a linear-optical setting. 
More generally as well, due to the Eastin--Knill theorem~\cite{Eastin2009}, transverse operators are also considered to be the ``easy'' set of operations on a QEC code as opposed to non-transversal operations. 
We consider a two-qubit logical gate to be transverse if it can be decomposed onto two-qubit physical operators each of which acts on one physical qubit from one code and its counterpart from the second code.

Note that in Fig.~\ref{fig_transverse}, we use a more restricted abstract measurement apparatus: contrary to the case of Fig~\ref{fig_logical} where the measurement apparatus can recover the measurement outcome in any case where it is theoretically possible, in Fig.~\ref{fig_transverse}, the measurement apparatus has restricted capabilities and can only recover the measurement outcome only in the transverse case. 
For logical BSMs, we should thus consider families of codes for which the logical $X$ and the logical $Z$ operators are transverse, which is for example the case for all the Calderbank--Shor--Steane codes.

\subsection{Logical LOBSM based on adaptive LOBSM.} \label{sec_adapt_LOBSM}

Using our formalism, we can derive the loss threshold for logical BSMs based on adaptive LOBSMs. This result relates to one of the main results from Ref.~\cite{Lee2019}, which claims to have derived the fundamental loss-tolerance threshold for linear optical logical BSMs. We will moderate this claim by proving that this is only true in the specific context of logical BSMs based solely on adaptive LOBSMs. Moreover, we will prove in Sec.~\ref{subsec_fundamental_limit_lobsm} that we can overcome this limit in a more general setting.

\begin{theorem} \label{th_absm}
    For a logical BSM based on adaptive LOBSMs acting between $a$ and $b$, the loss thresholds for qubit $a$, $\varepsilon_{a}^{(ABSM)}$, and $b$, $\varepsilon_{b}^{(ABSM)}$, are always bounded by:
    $$ \forall \code, \left(1 - \varepsilon_{a}^{(ABSM)}\right) \left(1 - \varepsilon_{b}^{(ABSM)}\right) \geq \sfrac12,$$
    where we use ``$(ABSM)$'' to indicate that we are restricting ourselves to decoders using adaptive physical LOBSMs.
\end{theorem}
\begin{figure}[!tb]
    \centering
    \includegraphics[width=0.9 \columnwidth]{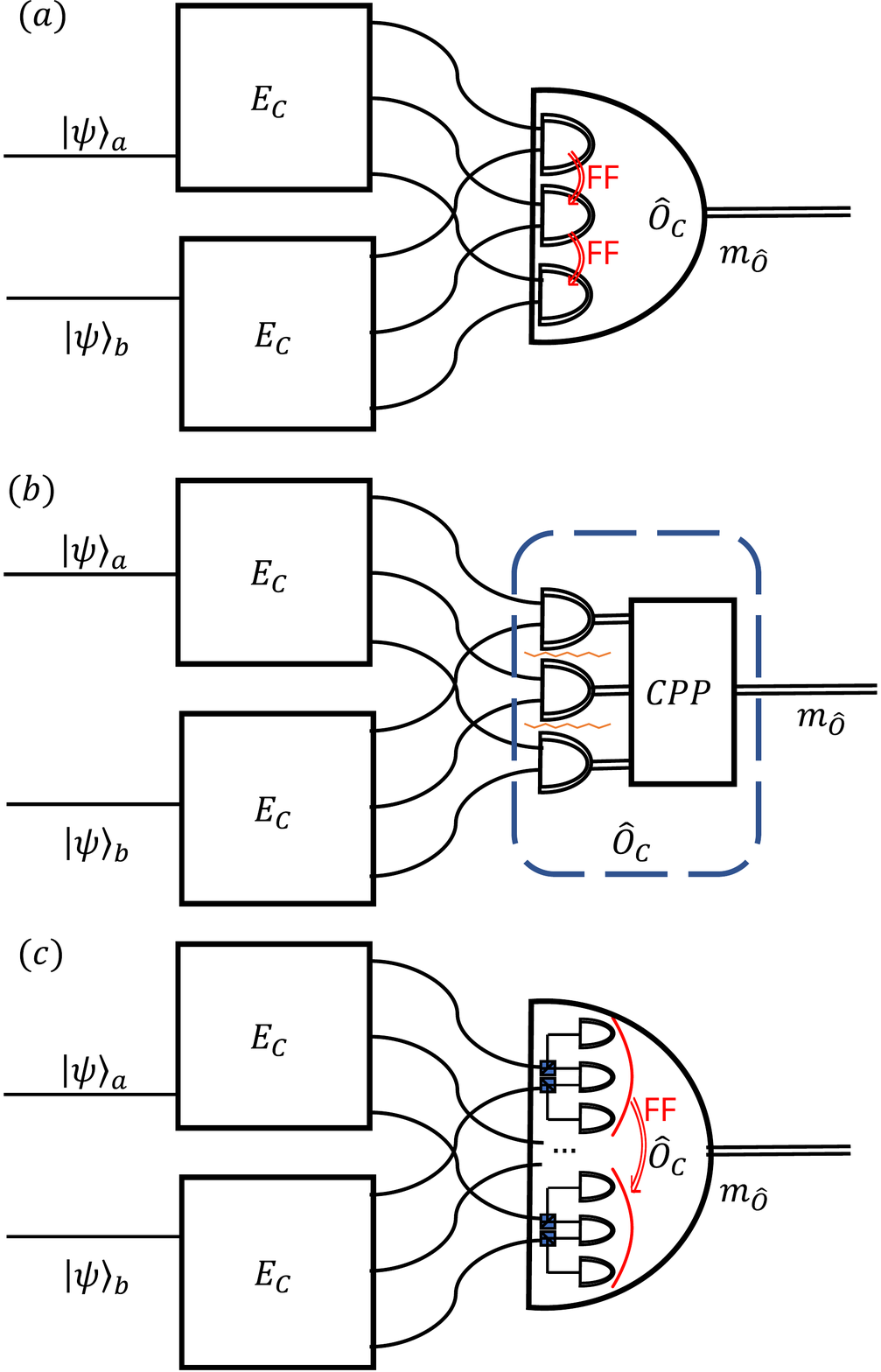}
    \caption{Logical measurement apparatus based on two-qubit linear-optical measurements. Based on (a) adaptive  and  (b) passive two-photon measurements and on an (c) adaptive combination of two-photon and single-photon measurements. CPP indicates a classical post-processing device that reconstructs the logical $\hat O$ measurement outcome from the measurement outcomes of each two-photon linear-optical measurement apparatus. FF stands for classical feed-forward.}
    \label{fig_LLOBSM}
\end{figure}
\begin{proof}
    We now impose restrictions on the logical measurement devices which should be solely based on physical LOBSM apparata, as depicted in Fig.~\ref{fig_LLOBSM}(a). Because we are using only physical LOBSMs, this logical BSM device works with the same probability if all the losses are transferred to the party $a$: $\eta_a' = \eta_a \eta_b$ and $\eta_b'=1$.
    Restricting ourselves to such logical measurement devices, we can derive a new bound using the measurement apparatus depicted in Fig.~\ref{fig_logical_Oa} with the specific operator measurement depicted in Fig.~\ref{fig_LLOBSM}(a).
    In that case, because $\eta_a' = \eta_a \eta_b = (1 - \varepsilon_a)(1 - \varepsilon_b)$, the measurement should be loss-tolerant if  $\varepsilon_a' = 1 - \eta_a'$ is below $\bar \varepsilon_a^{(BSM)}$ 
    which is itself bounded by $\sfrac12$ by Eq.~\eqref{eq_a_bar_eps_a}. 
    Therefore, we find the following condition for the thresholds $\varepsilon_{{a}}^{(ABSM)}$, and $\varepsilon_{{b}}^{(ABSM)}$:
    \begin{equation*}
         (1 - \varepsilon_{{a}}^{(ABSM)}) (1 - \varepsilon_{{b}}^{(ABSM)})
         \geq 1 - \bar\varepsilon_a^{(BSM)} \geq \sfrac12. 
    \end{equation*}
\end{proof}

Interestingly, to derive this bound, we only used the constraint \ref{LOBSMc1} of the physical LOBSM, but didn't consider its probabilistic nature (constraints \ref{LOBSMc2} and \ref{LOBSMc3}). In~\cite{Lee2019}, Lee et al. showed that this threshold is tight with an example of an \emph{adaptive} logical LOBSM onto quantum parity codes which reaches this loss tolerance. This result illustrates that we can handle the intrinsic probabilistic nature of the physical LOBSMs, by using adaptive measurements, i.e. by changing the LOBSM configuration and deciding which Bell states can be measured unambiguously. A second consequence of this result is that having access to ancillary photons \cite{Grice2011, Ewert2014, Wein2016, Olivo2018}, even with unit detection probability,  to allow physical LOBSMs with potentially near-deterministic success probability cannot improve the loss threshold of logical BSMs based on adaptive LOBSMs, even though it could still potentially improve performances in other aspects (e.g.\@ the number of physical qubits required in the code to reach a given success probability).

\subsection{Derivation of the logical LOBSM loss-threshold based on static LOBSMs} \label{subsec_static_LOBSM}
Another important consideration for applications is what is the best achievable loss-tolerance for static LOBSMs, where the physical LOBSM basis is chosen beforehand and is not modified based on the previous measurement outcomes,
which simplifies the experimental implementation.
This configuration corresponds to the one depicted in Fig.~\ref{fig_LLOBSM}(b), where there are no communications between the different physical LOBSMs and the resulting logical BSM is post-processed through the two-photon measurement outcomes of each BSM.

Here, we assume that we chose the same configuration for all of the physical LOBSMs, which is equivalent to choosing $p_{XX}$ and $p_{ZZ}$ constant for all the physical LOBSMs. In that case, the reasoning is similar to the adaptive case discussed in the previous section, except that we should fulfill the conditions

\begin{multline}
     (1 - \varepsilon_{a}^{(SBSM)}) (1 - \varepsilon_{b}^{(SBSM)}) \\
     \geq \max\left[\frac{(1 -\varepsilon_{X_a}^{\{X_a, Z_a\}})}{p_{XX}}, \frac{(1 -\varepsilon_{Z_a}^{\{X_a, Z_a\}})}{p_{ZZ}}\right]  \label{eq_absm_det}\\  \geq \sfrac23.     
\end{multline}

Indeed, using the same reasoning as for the previous case, we have $\varepsilon_a' = 1 - \eta_a'$ with $\eta_a' = (1-\varepsilon_a) (1-\varepsilon_b)$. We can use the setup illustrated in Fig.~\ref{fig_logical_Oa} but using the measurement apparatus described in Fig.~\ref{fig_LLOBSM}(b) to perform an indirect $X_a$ measurement from a $X_a X_b$ measurement. This measurement will succeed as long as 
$\eta_a' p_{XX} > 1 - \varepsilon_{X_a}^{\{X_a, Z_a\}}$.
We find a similar result for a $Z_a$ measurement. So $\eta_a'$ should be 
greater than both $\left(1 -\varepsilon_{X_a}^{\{X_a, Z_a\}}\right) / p_{XX}$ and $\left(1 - \varepsilon_{Z_a}^{\{X_a, Z_a\}}\right) / p_{ZZ}$, hence the first inequality in Eq.~\eqref{eq_absm_det}.

The final lower bound of $\sfrac23$ is found by minimizing this quantity under the linear-optic BSM constraints ($p_{XX} + p_{ZZ} \leq \sfrac32$) and Eq.~\eqref{eq_measurement_comm}. This bound is also tight because we can reach it with surface codes, by using BSMs with random bases $p_{XX} = p_{ZZ} = \sfrac34$. In that case, we can make a logical $Z_a Z_b$ (respectively $X_a X_b$) measurement on two surface codes if $p_{ZZ} \eta_a \eta_b$  (resp. $p_{XX} \eta_a \eta_b$) is above the percolation threshold $\sfrac12$, i.e.\@ if $\eta_a\eta_b \geq \sfrac23$. We have also found this bound numerically with a tree graph state logical encoding in a previous work~\cite{Hilaire2021error}. 
With access to ancillary state-assisted LOBSMs with overall success probability $\sfrac12 \leq p < 1$, this results straightforwardly 
generalizes to a tight lower bound of $1/(1+p)$.

\begin{theorem}
    For a logical BSM based on static LOBSMs, the loss-thresholds for qubits $a$, $\varepsilon_{a}^{(SBSM)}$, and $b$, $\varepsilon_{b}^{(SBSM)}$, are bounded by:
    \begin{equation}
    \forall \code, (1 - \varepsilon_{a}^{(SBSM)}) (1 - \varepsilon_{b}^{(SBSM)}) \geq \sfrac{1}{1+p},
    \label{eq_lobsm_static_p}
    \end{equation}
    when using LOBSM with success probability $p$.
\end{theorem}
\begin{proof}
    Appendix~\ref{app_static_BSM} contains the general proof of this result, including removing the assumption of taking identical $p_{ZZ}$ and $p_{XX}$ for any physical qubit. 
\end{proof}

The limit case, $p\to 1$, corresponds to a deterministic LOBSM, which also approaches the previous limit for adaptive logical BSMs based only on LOBSMs, thus showing that this limit can also be approached using ancillary state-assisted methods, though at the cost of a large overhead of ancillary photons consumed to perform near-deterministic physical LOBSMs.

\section{Best achievable loss tolerance for BSM with linear optics}
\label{subsec_fundamental_limit_lobsm}

We have shown in Sec.~\ref{sec_BSM} that the fundamental best achievable loss tolerance for logical BSMs is $\sfrac12$ ($\bar \varepsilon^{(BSM)} \leq \sfrac12$ for both $a$ and $b$)  without assuming any specific implementations, and we have seen that we can reach a lower loss-tolerance of $(1 - \varepsilon_a^{(ABSM)}) (1 - \varepsilon_b^{(ABSM)}) > \sfrac12$
using adaptive LOBSM (Th.~\ref{th_absm}).
Can we devise another way of making logical BSMs compatible with linear optics but which reach the fundamental bound? In the following, we show this to be the case.
We will provide a simple example of a logical LOBSM scheme which maximizes the amount of loss tolerated by the laws of physics.

Before doing so, we need first to understand why logical BSMs based on physical LOBSMs fail to reach the fundamental bounds for logical BSMs. Since we are making measurements of the photonic qubits two-by-two using LOBSMs, the measurement of each physical qubit is also conditioned on the detection of its counterpart from the second QECC, hence the appearance of the $(1 - \varepsilon_a) (1 - \varepsilon_b) > \sfrac12$ threshold. Contrarily, measurements of logical $X$ and $Z$ operators are only limited by the $\varepsilon_i < \sfrac12$ threshold ($\forall i \in \{a, b\}$), because they can be, for example, based on single-qubit measurements. The question is therefore whether we can use single-qubit measurements to perform a logical LOBSM.

It is easy to see that this is impossible if we are uniquely using single-qubit measurements, because we would need to measure the logical operators $X$ and $Z$ individually on each code, which is impossible because they are not commuting. However, nothing contradicts the idea of using logical BSMs based on an adaptive combination of both physical LOBSMs and single-qubit measurements as illustrated in Fig.~\ref{fig_LLOBSM}(c), which uses a physical measurement apparatus of the form in Fig.~\ref{fig_LOBSM}(b). The critical idea behind this is that we use some LOBSMs because they are needed for logical BSM measurements. Then, we use single-qubit measurements whenever possible to achieve better loss tolerance. Using this strategy it should be possible to perform LOBSMs with the maximum loss tolerance allowed by the laws of physics, as we illustrate by proposing a simple example in the following.

\begin{theorem}
    The tight fundamental upper-bound for logical BSMs of Th.~\ref{th_bsm} constitutes also a tight upper bound in a linear-optical setting:
    $$\forall \code, \varepsilon^{(LOBSM)} \leq \sfrac12.$$
\end{theorem}
We will prove this result by finding an example of a logical BSM reaching this loss tolerance based on linear optics.

\subsection{Example of logical BSM with maximum loss-tolerance}
\label{subsec:optimalcodeexample}

\begin{figure}[tb]
    \centering
    \includegraphics[width=0.8\columnwidth]{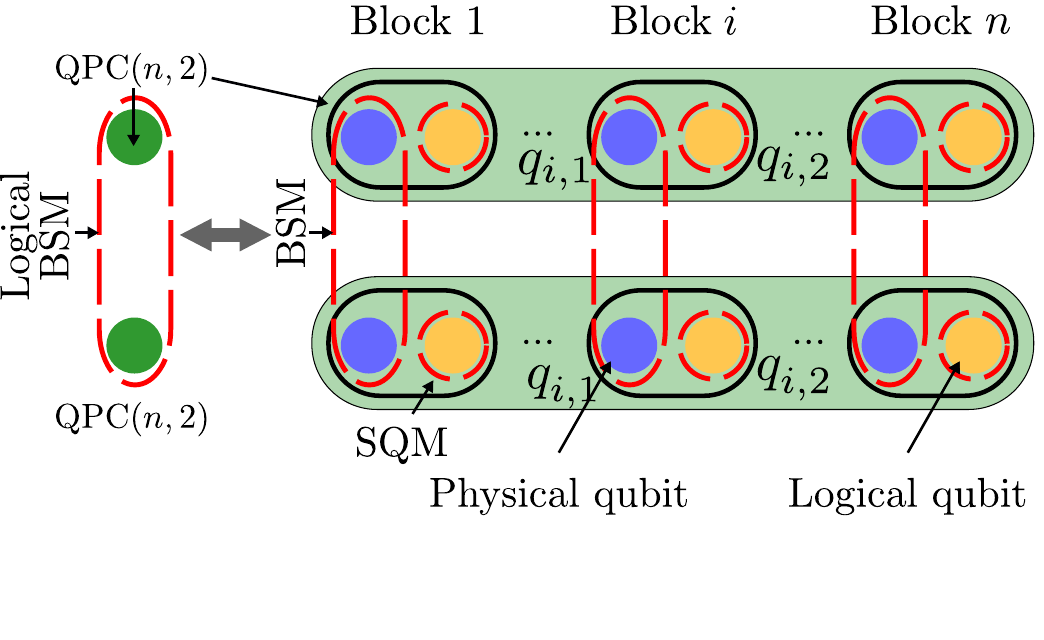}
    \caption{Variant of a QPC$(n,2)$ code with a logical LOBSM.}
    \label{fig_example}
\end{figure}

In this section, we propose a class of codes together with a scheme for logical LOBSMs based on adaptive LOBSMs and single-qubit measurements which reach the best achievable loss tolerance. We should emphasize that our objective in this section is to prove the tightness of the fundamental upper bound of logical LOBSMs. Therefore, we use this class of codes mostly as a theoretical tool to prove this tightness analytically. It also serves pedagogical purposes as it illustrates with a simple example how we can perform a loss-tolerant logical LOBSM based on a QECC and a combination of optical two-qubit and single-qubit measurements. 
Nevertheless, this code is inefficient in terms of resources required and is of very limited practical utility for implementations. Yet, we expect a similar approach with other codes to produce the same threshold  while being practical. For example, numerical evidence shows that tree graph codes~\cite{Hilaire2021error} have the same loss tolerance thresholds.

Our example is based on a variant of the quantum parity code QPC($n,m$) consisting of $n$ blocks of $m$ qubits each~\cite{Ewert2016}. Figure~\ref{fig_example} illustrates such a code. Quantum parity codes are a generalization of the well-known Shor code~\cite{Shor1995}, which is the special case corresponding to QPC($3,3$). In particular, we will focus on QPC($n,2$) composed of $n$ blocks of $2$ qubits, where, for each block $i = 1,..., n$, the first qubit $q^{}_{1,i}$ is simply a physical photonic qubit, and the second qubit $q^{}_{2,i}$ is itself encoded logically using a QECC with loss-tolerance thresholds for $X$ and $Z$ of $\varepsilon_{X}^{\{X, Z\}} = \varepsilon_{Z}^{\{X, Z\}} = \sfrac12$, such as surface codes or tree graph states~\cite{Varnava2006}. We want to show that we can construct a logical LOBSM on two logically-encoded qubits, based on this variant of the QPC($n,2$), whose loss-tolerance threshold is exactly the one of the single-qubit QECC used for the qubits $q^{}_{2, i}$.

For a quantum parity code, a logical $X$ measurement corresponds to the successful $X$-measurement of all qubits from at least one single block, 
and a logical $Z$ measurement corresponds to the $Z$ measurement of at least one qubit in each block.
A logical BSM therefore corresponds to the measurement of the logical operators $X_{a} X_{b}$ and  $Z_{a} Z_{b}$:
\begin{gather} 
    \forall i,j=1,...,n, \quad  X_{a} X_{b} = X^{a}_{i,1}X_{j,1}^{b} X^{a}_{i,2}X_{j,2}^{b}  \\
    \forall k_1, ..., k_n, l_1, ..., l_n, \in \{1,2\},\quad Z_{a} Z_{b} = \prod_{j=1}^n Z^{a}_{j,k_j} Z_{j,l_j}^b,
    \label{eq_bsm}
\end{gather}
where $Z^{a}_{i,j}$ and $X^{a}_{i,j}$ (respectively $Z_{i,j}^b$ and $X_{i,j}^b$) correspond to $Z$ and $X$ operators on the qubit $q^{a}_{i,j}$ (resp. $q^{b}_{i,j}$) from logical qubit $a$ (resp. $b$).
The logical BSM with record loss-tolerance is illustrated in Fig.~\ref{fig_example} and detailed in the following.

On each block, $q^{a}_{1,n}$ and $q^{b}_{1,n}$ are measured jointly with a physical LOBSM. $q^{a}_{2,n}$ and $q^{b}_{2,n}$ are both measured individually in the $X$ basis if the BSM on $q^{a}_{1,n}$ and $q_{1,n}^{b}$ has succeeded;  
otherwise  they are both measured the $Z$ basis. A logical BSM then succeeds if the following two conditions are satisfied:
\begin{enumerate}
    \item at least one BSM succeeds, \label{logBSMc1}
    \item all the single-qubit measurements succeed. \label{logBSMc2}
\end{enumerate}
Indeed, as shown in Eq.~\eqref{eq_bsm}, supposing that both the BSM on $q^{a}_{1,i}$ and $q_{1,i}^{b}$ and the $X$ measurements on $q^{a}_{2,i}$ and $q_{2,i}^{b}$ have succeeded, allows us to retrieve $X_{a} X_b$ (through the $X$ measurements of all the qubits from a block). 
Moreover, from the BSM on $q^{a}_{1,i}$ and $q_{1,i}^{b}$, we have access to $Z^{a}_{1,i} Z_{1,i}^b$ too. With that information and with all the individual $Z$ measurements on all the other blocks (or the joint $ZZ$ measurements if the corresponding BSM has succeeded), we can thus retrieve also $Z_{a} Z_b$ and thus implement a complete logical BSM.

At least one out of the $n$ BSMs succeeds (condition \ref{logBSMc1}) with probability $1 - (1 - (1-\varepsilon_a) (1-\varepsilon_b) / 2)^n$, which can be made arbitrarily close to 1 by increasing the number of blocks $n$. In addition, we use an encoding with sufficiently strong loss tolerance so that all the $2n$ single-qubit measurements succeed (condition \ref{logBSMc2}) with arbitrarily high probability. This is possible as long as $\varepsilon_i<\min[\varepsilon_{X}^{\{X, Z\}}, \varepsilon_{Z}^{\{X, Z\}}]$ for $i = a, b$. Since QECCs fulfilling this condition exists with a loss threshold $\sfrac12$, we conclude that by opting for this encoding, we obtain a logical LOBSM with maximum loss-tolerance reaching $\sfrac12$.

\subsection{Loss-tolerant linear-optical decoder}

We show now that we can design a linear-optical decoder with maximum loss-tolerance $\varepsilon_\decode = \sfrac12$, bounded by the fundamental limits of Th.~\ref{th_no_cloning_bound}.
The architecture of this decoder is based on quantum teleportation and is already known, it was for example investigated in~\cite{Ewert2016, Lee2019}. However, given the new results from the previous subsection, it can now operates for a larger amount of internal losses in the decoder.

\begin{corollary}
    $\forall \code$, the loss-tolerance of a linear-optical decoder, $\varepsilon_\decode^{(LO)}$, has a tight upper bound of $\sfrac12$.
\end{corollary}

\begin{figure}[!ht]
    \centering
    \includegraphics[width=1\columnwidth]{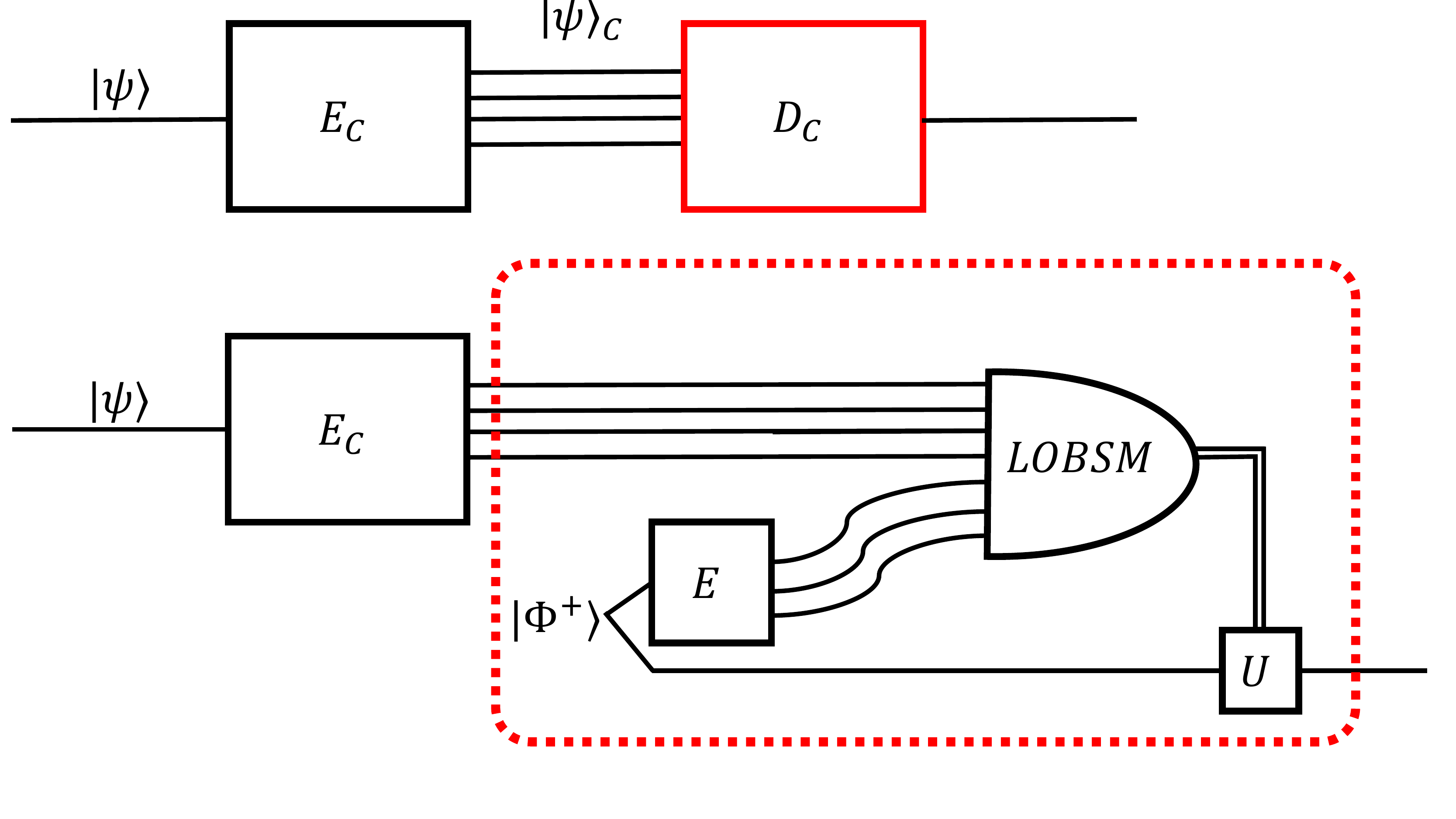}
    \caption{A linear-optical decoder based on quantum teleportation. The logical LOBSM can be any of the one depicted in Fig.~\ref{fig_LLOBSM}. The unitary $U$ corresponds to the corrections that need to be applied to the qubit depending on the LOBSM outcome.}
    \label{fig_LO_decoder}
\end{figure}

\begin{proof}

Since we are processing quantum information only using linear-optical components and detectors, i.e.\@ only destructive photon measurements and probabilistic two-qubit gates, designing such a decoder is not as straightforward as without the linear-optical constraints.

To build such a linear-optical decoder, we need not only to measure the state (with destructive measurements) but also to recover it. This is typically how quantum teleportation works. We consider a quantum state $\ket{\psi}_{a}$ embedded on a physical qubit $a$, and two other qubits $b$ and $c$ prepared in a Bell state $\ket{\Phi^+}_{bc}$. By performing a physical BSM onto qubits $a$ and $b$, the quantum state of qubit $c$ after measurement is projected into the quantum state $\ket{\psi}_c$, up to some known single-qubit gates depending on the BSM outcome. Here, we see that in a linear-optical setting, the qubit $a$, which was the initial support of the quantum state $\ket{\psi}$ is not existing anymore but the quantum state has been maintained and transferred onto qubit $b$ during the quantum teleportation.

However, acting on physical qubits, the quantum teleportation scheme described previously is not loss-tolerant. Yet, we can easily convert it into a loss-tolerant linear-optical decoder by encoding qubits $a$ and $b$ logically and by replacing the physical BSM with a logical LOBSM. This is illustrated in Fig.~\ref{fig_LO_decoder}. The success of this decoder depends on the success of the logical LOBSM, which can be performed with loss tolerance of $\sfrac12$. Therefore, the upper bound for linear-optical decoder is also $\sfrac12$. Note that this bound is tight since we have already found a logical LOBSM with similar tight upper bound.
\end{proof}

We should also note that this is at the core of most all-photonic quantum repeater schemes~\cite{Ewert2016, Lee2019, Hilaire2021error, Niu2022, Bell2022}. These protocols are based on performing LOBSMs on logical Bell states to propagate a quantum state through a lossy channel. In that case, we just need to also logically encode qubit $c$ in our linear-optical decoder.

Moreover, with our framework and our proposal for a linear-optical decoder based on quantum teleportation, we straightforwardly find lower bounds for the best logical LOBSMs solely based on adaptive or static physical LOBSMs.

\begin{corollary}
    $\forall \code$, the fundamental limits for linear-optical decoders based on adaptive and static physical LOBSMs are respectively:
    $$
    \varepsilon_\decode^{(ABSM)} \leq \sfrac12, \; {\rm and} \; 
    \varepsilon_\decode^{(SBSM)} \leq \sfrac{1}{1+p}.
    $$
\end{corollary}
Interestingly, we can find codes for which $\varepsilon_\decode^{(ABSM)} = \varepsilon_\decode^{(LO)}$, this is a consequence of the fact that a physical LOBSM on qubit $a$ and $b$ works with the same probability for transmissions $\eta_a, \eta_b$ and for transmission $\eta_a'=\eta_a \eta_b, \eta_b'=1$ (see Fig.~\ref{fig_LOBSM}).

\section{Overview of the results} \label{sec_overview}

In the previous sections, we have proposed a framework to find fundamental bounds on the overall loss-tolerance thresholds and the loss-tolerance of some operator measurements for general codes. We applied this framework to investigate the fundamental loss-tolerance bounds of BSMs on logical qubits. We have first derived that while it is possible, in principle, to perform a logical BSM in the presence of arbitrarily high qubit loss on qubits within the same QECC, the codes where this is possible may not be of practical interest. Indeed, these QECCs do not exhibit loss tolerance for decoding and thus a quantum state cannot be protected against loss using such a QECC. Moreover, the BSM should act on qubits which are logically encoded within the same QECC and thus cannot be prepared remotely, potentially reducing their interest e.g.\@ for long-distance communications.

We then proceeded by showing that BSMs on logical qubits encoded remotely have a fundamental loss-tolerance of at best $\sfrac12$, corresponding to the fundamental limit imposed by the measurement postulate. Furthermore, we derived loss-thresholds in different contexts where we focus on photonic implementation and linear-optical quantum information processing. We showed that:
\begin{itemize}
\item for static logical LOBSMs, the loss probabilities, $\varepsilon_a, \varepsilon_b,$ of physical qubits within each code $a$ and $b$ should meet the condition $(1 - \varepsilon_a) (1 - \varepsilon_b) > \sfrac23$,
\item for logical BSMs based on adaptive physical LOBSMs, the loss-tolerance condition increases to $(1 - \varepsilon_a) (1 - \varepsilon_b) > \sfrac12$,
\item for logical BSMs based on adaptive physical LOBSMs and single-qubit measurements, we reach the fundamental bounds of logical BSMs: $\max[\varepsilon_a, \varepsilon_b]< \sfrac12$.
\end{itemize}

Furthermore, we have shown the tightness of each of these bounds. In particular, the last bound may look surprising given that Lee et al.~\cite{Lee2019} stated that the best achievable loss tolerance for LOBSMs should be $(1 - \varepsilon_a) (1 - \varepsilon_b)>\sfrac12$. Our framework helps to understand why this limit can be overcome by highlighting that the authors made one hypothesis for the derivation of this result which turned out to be incomplete. This hypothesis was that a logical LOBSM should always be decomposed onto physical LOBSMs. We provide a simple counter-example which shows that this is not always the case, and that a better loss-tolerance is achievable.
An overview of all the results is displayed in Table.~\ref{table_res} and in Fig.~\ref{fig_perf} (a). Moreover, we have also shown that the results in terms of loss-tolerance for logical BSMs directly translate into a loss-tolerant decoder with similar performances, using a logical quantum teleportation scheme.

\begin{figure*}[tb]
    \centering
    \includegraphics[width=1.7\columnwidth]{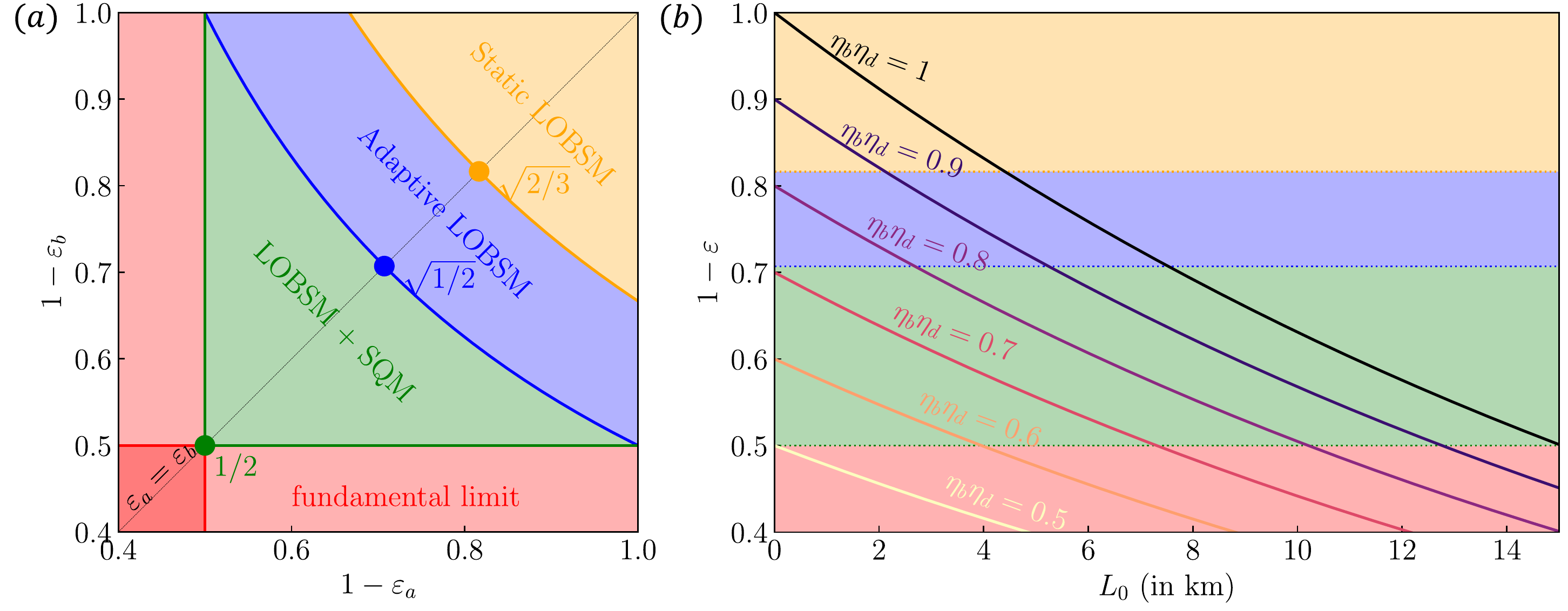}
    \caption{(a) Summary of the fundamental bounds for logical BSMs in a linear-optical setting. The best loss tolerance is achieved for schemes using LOBSM and single-qubit measurements (SQM). (b) overall single photon detection probability $1 - \varepsilon$ for a fiber of distance $L_0$ (in km) given some source ($\eta_b$) and detection ($\eta_d$) efficiencies. Fundamental bounds are taken for $\varepsilon_a = \varepsilon_b = \varepsilon$.}
    \label{fig_perf}
\end{figure*}

\begin{table}[tbp]
    \centering
    \begin{ruledtabular}
    \begin{tabular}{lccc}
        \multicolumn{1}{c}{\multirow{2}{4em}{}} & \textbf{Static} & \multicolumn{2}{c}{\textbf{Adaptive \phantom{xx}}} \\
        \textbf{Physical BSMs}& \textbf{BSM} & \textbf{BSM} & \textbf{BSM + SQM} \\
         \hline
        LOBSM ($p=\sfrac12$)& $1 - \sqrt{\sfrac{2}{3}}$ & $1 -\sfrac{1}{\sqrt{2}}$ & $\sfrac{1}{2}$ \\
        LOBSM (assisted) & $1 -\sfrac{1}{\sqrt{1 + p}}$ & $1 - \sfrac{1}{\sqrt{2}}$ & $\sfrac{1}{2}$ \\
        LOBSM (assisted $p\to1$) & $1 - \sfrac{1}{\sqrt{2}}$ & $1 - \sfrac{1}{\sqrt{2}}$ & $\sfrac{1}{2}$ \\
        Deterministic  & $\sfrac{1}{2}$ & $\sfrac{1}{2}$ & $\sfrac{1}{2}$ \\
        \end{tabular}
        \end{ruledtabular}
        \caption{Table of results (for equal losses $\varepsilon_a = \varepsilon_b$). Fundamental loss thresholds using standard LOBSMs, ancillary-state-assisted LOBSMs (with success probability $p$) and for deterministic BSMs (not using linear-optics). We consider the case where the logical measurements are performed using a static or an adaptive protocol and in the case of adaptive measurements, when we perform either BSMs only or a combination of BSMs and single-qubit measurements (SQM).}
    \label{table_res}
\end{table}

We now discuss the importance of these fundamental bounds. 
Improving the loss-tolerance of QECCs, and in particular logical BSMs is particularly important for photonic implementations where losses are the main source of errors.

Our results clarify the requirements a logical BSM scheme should fulfill to reach optimal loss tolerance using linear optics. In particular, we show that such a logical BSM should first exploit feed-forward. Otherwise, its loss tolerance is limited by the orange region in Fig.~\ref{fig_perf}(a). Moreover, a logical BSM with maximal loss tolerance should also use both physical LOBSMs and single-qubit measurements. Otherwise, its performance is necessarily limited to the blue and orange regions in Fig.~\ref{fig_perf}(a). We hope highlighting these two necessary conditions will help devise better photonic logical BSMs.
Moreover, the loss-tolerance threshold can be seen as the ``loss budget'' for a practical photonic implementation. The overall photon transmission should include every operation from its generation, its use, and its detection. Therefore, improving the loss-tolerance is crucial to increase the number of operations that we can perform onto a photon before it should be detected.
The new logical LOBSM loss-threshold is thus a crucial improvement since the brightness $\eta_b$ of state-of-the-art single-photon sources~\cite{Kaneda2019, Langenfeld2020, Bhaskar2020, Tomm2021, Thomas2021} and the single-photon detector efficiency $\eta_d$ are sufficiently high to reach this new regime of loss tolerance but not the previous one: $1/2 < \eta_b \eta_d < 1 / \sqrt{2}$.

For quantum communication applications, this result could improve the implementation of quantum repeater protocols which aim at enabling long-distance quantum communications. Indeed, the most advanced quantum repeater schemes are based on quantum error correction, including the newly investigated all-photonic quantum repeater protocols~\cite{Azuma2015, Ewert2016, Lee2019,Hilaire2021error,Hilaire2021resource, Zhang2022, Niu2022, Bell2022}. For these QECCs, the principal objective is to enable communications through a lossy channel, such as a telecom fiber, with typical attenuation of $0.2$dB/km, 
leading to losses $\sfrac12$ after a distance $L_{\sfrac12}\simeq \sfrac{3\ \mathrm{dB}}{0.2\ \mathrm{dB}\cdot\mathrm{km}^{-1}} = 15\ \mathrm{km}$. 
The loss-tolerance threshold of the QECC used imposes a lower bound on the distance between two repeater nodes. Indeed, if the overall single-photon detection loss (including generation, fiber transmission, and detection) is above the BSM loss threshold the quantum repeater scheme cannot work. As a result, a quantum repeater scheme based on a QECC with loss threshold $\varepsilon_{BSM}$, can work only if the fiber transmission efficiency $\eta_t(L)$ is above $(1 - \varepsilon_{BSM}) / (\eta_b \eta_d)$.
In Fig.~\ref{fig_perf}(b), we show how stringent this requirement can be for imperfect photon sources and detectors ($\eta_b \eta_d < 1$). For example, if $\eta_b \eta_d = 0.8$, a repeater scheme using logical BSMs based on:
\begin{itemize}
    \item static LOBSMs cannot work.
    \item adaptive physical LOBSMs requires a very small maximum internode distance of $L_O\approx 2$km.
    \item adaptive physical LOBSMs and single-qubit measurements requires a greater maximum internode distance of  $\approx 10$km.
\end{itemize}
The latter is thus much more practical for implementations.
Here, we focused on quantum communication applications because the internode distance of quantum repeater schemes is a simple and intuitive number to maximize. However, we expect this result to be also of practical interest for quantum computing applications, particularly for FBQC~\cite{Bartolucci2021}. Indeed, any operations made using photonic integrated circuits are lossy, including delay lines. Therefore, having greater tolerance to losses should increase the number of operations that can be performed on a qubit before it is measured.

We should emphasize that the results that we derived for linear-optics quantum information processing were obtained with a simplified description where we have only considered single-qubit measurements and probabilistic BSMs.
By considering more general linear-optical operations, better loss-tolerance thresholds may potentially be found. However, the upper bound of $\sfrac12$ on linear-optical logical BSMs and decoders reach the more general upper bounds imposed by quantum mechanics, and thus cannot be improved.
Given the proof of the upper bound for adaptive logical BSMs based only on LOBSMs, $\varepsilon_{BSM}^{(ABSM)}$, in~\cite{Lee2019} shows that this upper bound is also fundamental. The only improvement that may be found is in the static LOBSM setting if we allow more general linear-optical information processing. Proving whether this would be the case or not could be an interesting extension of this work. More generally, having a deeper understanding of the linear-optical methods to process photonic qubits would always be useful to find efficient ways to process quantum information encoded onto photonic qubits, not only in terms of loss thresholds but also in terms of the amount of resources used.

Furthermore, the framework that we derived could be extremely useful in itself if we can extend it to other types of errors such as operation errors. This could greatly help the derivation of better linear-optical schemes for quantum information processing. Another aspect that is omitted here but which is of important practical interest is an estimation of the resource overhead, i.e. the number of physical qubits used in these codes. 
Indeed, the current analysis only deals with loss threshold but doesn't say how many qubits need to be used to reach a satisfactory success rate. It would be also extremely valuable to investigate in a future work the resource overhead induced by the probabilistic LOBSMs compared to deterministic schemes which are not based on linear-optics.

\begin{acknowledgments}
    We thank Shane Mansfield for interesting and fruitful discussions.
    FG and PH acknowledge funding from the Plan France 2030 through the ANR-22-PETQ-0006 NISQ2LSQ project.
    FG acknowledges support of the ANR through the ANR-17-CE24-0035 VanQute project.
    PH is grateful for support from the grant BPI France Concours Innovation PIA3 projects DOS0148634/00 and DOS0148633/00 – Reconfigurable Optical Quantum Computing.
    PH and SEE acknowledge support by the EU Horizon 2020 programme (GA 862035 QLUSTER).
    SEE also acknowledges the Virginia Commonwealth Cyber Initiative (CCI).
    EB acknowledges NSF grant no. 2137953, QuIC-TAQS.  
\end{acknowledgments}

\appendix

\section{Proof of the $\sfrac{1}{1+p}$ threshold for static LOBSM}~\label{app_static_BSM}

To prove this result we consider the case where a physical LOBSM succeed with intrinsic probability $p$ (multiplied by $\eta_a \eta_b$ to account for losses), therefore, the best success probability for physical $XX$ and $ZZ$ measurements on qubits labelled by $i$ are linked by $p_{XX,i} + p_{ZZ, i} \leq 1 + p$. To complete the proof of the main manuscript, we should also consider the case where the probabilities $p_{XX, i}$ and $p_{ZZ, i}$ vary for different pair of qubits $i$ being measured. Since we are considering static LOBSM, we should consider vectors $\vec p_{XX} = (p_{XX,1}, p_{XX,2}, ..., p_{XX,n})^T$ and $\vec p_{ZZ} = (p_{ZZ,1}, p_{ZZ,2}, ..., p_{ZZ,n})^T$ which follows the relation $\vec p_{XX} + \vec p_{ZZ} \leq (1 + p) \vec 1$. Including single qubit loss ($\eta_a = 1 - \varepsilon_a$, $\eta_b = 1 - \varepsilon_b$), we have the relation 
$$\vec p_{XX} \eta_a \eta_b + \vec p_{ZZ}  \eta_a \eta_b \leq (1 + p) \eta_a \eta_b \vec 1$$
where $\vec 1$ is just a vector of ones with length the number of qubit pair $n$.

Using the relations with Eq.~\eqref{eq_measurement_comm}, we find that $P(X_L X_L | \eta_a, \eta_b, \vec p_{XX}) = P(X_L | \vec p_{XX} \eta_a \eta_b)$ (and similarly for $Z$). In that case, we consider different probabilities for the measurement of each operator on each qubits by using this vector representation.
We can derive a similar relation for single logical qubit operator measurements: $P(X_L | \eta_a \eta_b \vec p_{XX}) + P(Z_L | \vec 1 - \eta_a \eta_b \vec p_{ZZ}) \leq 1$.

We show by contradiction the threshold for static linear optics with $p$ physical LOBSM success probability to be 
\begin{equation}
    (1 - \varepsilon_{A}^{(SBSM)}) (1 - \varepsilon_{B}^{(SBSM)}) \geq \sfrac{1}{1+p}.    
\end{equation}

If we have single photon detection probabilities following the condition $ \eta_a \eta_b \leq 1/ (1+p)$, then $\vec p_{ZZ} \eta_a \eta_b  \leq \vec 1 - \eta_a \eta_b \vec p_{XX}$ and thus the previous relation imply that
$$
P(X_L X_L | \eta_a, \eta_b, \vec p_{XX}) + P(Z_L Z_L | \eta_a, \eta_b, \vec p_{XX}) \leq 1
$$
Therefore, we cannot perform a complete logical BSM measurement, which conclude that the loss threshold is indeed bounded as in Eq.~\eqref{eq_lobsm_static_p}. Moreover, we can easily show that this limit is tight using the same example of a surface code as in the main manuscript using physical LOBSM with success probability $p$.

\bibliographystyle{ieeetr}
\bibliography{bib}

\end{document}